\newif\ifacm
\acmfalse

\ifacm
\documentclass[sigplan,screen]{acmart}
\else
\documentclass[sigplan,screen]{acmart}
\renewcommand\footnotetextcopyrightpermission[1]{}
\fi

\usepackage{hyperref}

\usepackage{latexsym,amsmath,amsfonts,stmaryrd,mathtools}
\usepackage{amsthm}
\usepackage{algorithm, algorithmicx}
\usepackage{listings, fancyvrb}
\usepackage{xspace}
\usepackage{titlecaps}
\usepackage{cleveref}

\usepackage[scaled=0.9]{beramono} 

\microtypecontext{spacing=nonfrench}

\newtheorem{theorem}{Theorem}[section]
\newtheorem{lemma}[theorem]{Lemma}

\newtheorem{definition}{Definition}
\newtheorem{observation}{Observation}
\newtheorem{assumption}{Assumption}

\makeatletter
\lst@Key{countblanklines}{true}[t]%
{\lstKV@SetIf{#1}\lst@ifcountblanklines}

\lst@AddToHook{OnEmptyLine}{%
	\lst@ifnumberblanklines\else%
	\lst@ifcountblanklines\else%
	\advance\c@lstnumber-\@ne\relax%
	\fi%
	\fi}
\makeatother

\Crefname{listing}{Algorithm}{Algorithms}

\def\ContinueLineNumber{\lstset{firstnumber=last}}
\def\StartLineAt#1{\lstset{firstnumber=#1}}

 \usepackage{titlesec} 
\titlespacing*\section{0pt}{6pt plus 4pt minus 2pt}{3pt plus 2pt minus 2pt}
\titlespacing*\subsection{0pt}{5pt plus 4pt minus 2pt}{3pt plus 2pt minus 2pt}
\titlespacing*\subsubsection{0pt}{5pt plus 4pt minus 2pt}{3pt plus 2pt minus 2pt}
\titlespacing*\paragraph{10pt}{2pt plus 1pt minus 1pt}{3pt plus 2pt minus 2pt}

\newcommand{\sysname}{SWARM\xspace}
\newcommand{\slogan}{Replicating Shared Disaggregated-Memory Data in No Time}
\newcommand{\protocol}{Safe-Guess\xspace}
\newcommand{\innout}{In-n-Out\xspace}
\newcommand{\kvs}{\sysname-KV\xspace}

\newcommand{\us}{µs\xspace}
\renewcommand{\t}[1]{\texttt{\footnotesize #1}}

\newcommand{\rt}{roundtrip\xspace}
\newcommand{\rts}{roundtrips\xspace}

\newcommand{\writeop}{\textsc{write}\xspace}
\newcommand{\readop}{\textsc{read}\xspace}
\newcommand{\casop}{\textsc{cas}\xspace}
\newcommand{\maxop}{\textsc{max}\xspace}

\newcommand{\getop}{\textsc{get}\xspace}
\newcommand{\updateop}{\textsc{update}\xspace}
\newcommand{\insertop}{\textsc{insert}\xspace}
\newcommand{\deleteop}{\textsc{delete}\xspace}
\newcommand{\wmr}{Max Register\xspace}
\newcommand{\false}{\t{False}\xspace}
\newcommand{\true}{\t{True}\xspace}

\newenvironment{myitemize}{\begin{list}{\labelitemi}{%
\setlength{\topsep}{0.5pt plus 0pt minus 0pt}%
\setlength{\itemsep}{0pt plus 0pt minus 0pt}%
\setlength{\parsep}{0pt plus 0pt minus 0pt}%
\setlength{\parskip}{0pt plus 0pt minus 0pt}%
}}{\end{list}}

\ifacm\else
\fi

\AtBeginDocument{%
  \providecommand\BibTeX{{%
    \normalfont B\kern-0.5em{\scshape i\kern-0.25em b}\kern-0.8em\TeX}}}

\ifacm
\acmYear{2024}\copyrightyear{2024}
\setcopyright{acmlicensed}
\acmConference[SOSP '24]{ACM SIGOPS 30th Symposium on Operating Systems Principles}{November 4--6, 2024}{Austin, TX, USA}
\acmBooktitle{ACM SIGOPS 30th Symposium on Operating Systems Principles (SOSP '24), November 4--6, 2024, Austin, TX, USA}
\acmDOI{10.1145/3694715.3695945}
\acmISBN{979-8-4007-1251-7/24/11}
\acmPrice{}
\else
\acmConference[SOSP '24]{ACM SIGOPS 30th Symposium on Operating Systems Principles}{November 4--6, 2024}{Austin, TX, USA [Preprint]}
\fi

\begin{document}

\title
[\sysname: \slogan]
{\LARGE \sysname: \slogan}


\author[Murat]{Antoine Murat}
\email{antoine.murat@epfl.ch}
\affiliation{
  \institution{École Polytechnique Fédérale de Lausanne (EPFL)}
  \country{Switzerland}
}
\authornote{These authors contributed equally to this work.}

\author[Burgelin]{Clément Burgelin}
\email{clement.burgelin@epfl.ch}
\affiliation{
  \institution{École Polytechnique Fédérale de Lausanne (EPFL)}
  \country{Switzerland}
}
\authornotemark[1]

\author[Xygkis]{Athanasios Xygkis}
\email{athanasios.xygkis@oracle.com}
\affiliation{
  \institution{Oracle Labs}
  \country{Switzerland}
}

\author[Zablotchi]{Igor Zablotchi}
\email{igor@mystenlabs.com}
\affiliation{
  \institution{Mysten Labs}
  \country{Switzerland}
}

\author[Aguilera]{Marcos K. Aguilera}
\email{marcos-k.aguilera@broadcom.com}
\affiliation{
  \institution{VMware Research Group}
  \country{United States}
}

\author[Guerraoui]{Rachid Guerraoui}
\email{rachid.guerraoui@epfl.ch}
\affiliation{
  \institution{École Polytechnique Fédérale de Lausanne (EPFL)}
  \country{Switzerland}
}

\begin{abstract}
Memory disaggregation is an emerging data center architecture
    that improves resource utilization and scalability.
Replication is key to ensure the
    fault tolerance
    of applications, but replicating shared data in
    disaggregated memory is hard.
We propose \sysname (Swift WAit-free Replication in disaggregated Memory),
    the first replication scheme for
    in-disaggregated-memory shared objects to provide
    (1) single-\rt \readop{}s and \writeop{}s in the common case,
    (2) strong consistency (linearizability), and
    (3) strong liveness (wait-freedom).
\sysname makes two independent contributions.
The first is \protocol, a novel wait-free replication protocol with
    single-\rt operations.
The second is \innout, a novel
    technique to provide conditional atomic update and atomic retrieval
    of large buffers in disaggregated memory in one \rt.
Using \sysname, we build \kvs,
    a low-latency, strongly consistent and highly available disaggregated key-value store.
We evaluate \kvs and find that it has marginal
    latency overhead compared to an unreplicated key-value store,
    and that it offers much lower latency and better availability
    than FUSEE, a state-of-the-art replicated disaggregated key-value
    store.
\end{abstract}

\maketitle

\section{Introduction}\label{sec:intro}

Memory disaggregation allows
  servers to access external memory provided by a set of
  memory nodes connected to a low-latency high-throughput fabric,
  using technologies such as RDMA~\cite{pfister2001infiniband} and, more recently, CXL~\cite{intelcxl}.
Memory disaggregation improves utilization of memory and hence decreases its relative
   cost~\cite{pond,zhou2022carbink,gu2017infiniswap}.
However, failures of the memory nodes can severely disrupt users and
  decrease overall system reliability.
To tolerate failures, traditional systems replicate data so that it remains accessible even
  if some of the replicas fail.
Unfortunately, current replication schemes are ill-suited for disaggregated memory,
because of two issues.
First, they require several network \rts to the memory nodes, which at least
   doubles the latency of disaggregated memory (\S\ref{sec:background:abd}). 
Second, they often require running code next to the data for concurrency control
  (e.g., to track timestamps, filter messages with old timestamps, issue a vote, etc.), which may be impossible or undesirable: the disaggregated memory may have no compute 
  capability (e.g., with CXL), or it may impose additional overhead to run code 
  (e.g., RDMA requires costly two-sided operations~\cite{kalia2019erpc, kalia2014herd}).
These issues are particularly problematic for systems requiring
 low latency, such as data stores in trading systems, control systems, and microsecond-scale microservices.

We propose \sysname
  (Swift WAit-free Replication in disaggregated Memory),
a new replication protocol for shared data in disaggregated memory.
\sysname provides several important features for our setting:
  it can read or write shared objects in one network \rt in the common case, when there are no failures or contention,
  and clocks are nearly synchronized;
  it can replicate both small and big objects;
  it provides (a) strong consistency in the form of linearizability~\cite{herlihy1990-linearizability},
  and (b) strong liveness in the form of wait-freedom~\cite{herlihy1991-waitfree}.

Designing \sysname required overcoming
multiple challenges.
First, the protocol must handle concurrent \readop{}s and \writeop{}s on the same object,
  while ensuring strong consistency,
  with little
  communication in the common case.
Second, the protocol cannot run computation in the memory nodes.
Third, the protocol must atomically handle \readop{}s and \writeop{}s of objects that are larger than the
  largest atomic updates supported by the disaggregated memory.

We address these challenges through a two-step modular design.
First, we design
  \protocol, a protocol
  that provides
  most of the properties we desire (single-\rt operations, linearizability,
  wait-freedom), but
  requires memory nodes to
  support a \emph{max register} object
that tracks the maximum value written to it out of large values (larger than the word size).
Internally, \protocol' \writeop{}s guess ordering timestamps speculatively to save a \rt; commonly, guesses are proven correct and \writeop{}s finish in one \rt; otherwise, \protocol uses a novel wait-free locking mechanism called \textit{timestamp locks} to resolve conflicts with potential readers, so that \writeop{}s can safely re-execute with better timestamps if needed.

Second, we introduce a novel technique, called \innout, that allows us to implement, on memory nodes with no compute, a
  max register for large values
  that takes a single
  \rt in the common case.
Doing so requires providing conditional updates
   on large values, but existing techniques
   take
   multiple \rts as
   they incur locking or pointer-chasing overheads to perform updates in-place
   or out-of-place, respectively.
Our \innout scheme achieves one \rt by simultaneously doing in-place and out-of-place updates:
   it uses in-place updates without locks to execute quickly when there is no contention, and
   it falls back to out-of-place updates when an in-place \readop{} finds corrupted data.
Furthermore, \innout optimizes out-of-place \writeop{}s to run in one \rt using 
    hardware ordering guarantees.

We combine \protocol and \innout to obtain \sysname.
To demonstrate the utility of \sysname, we build
\kvs,
    a strongly consistent and highly available disaggregated key-value store
    with ultra-low latency where clients directly access data on
    memory nodes.
\kvs can serve \insertop{}s, \updateop{}s, \getop{}s and \deleteop{}s in a single \rt
    and continue operating despite the failure of clients and memory nodes.
We evaluate \kvs against two main baselines:
    (1) a raw disaggregated key-value store that is unreplicated and has no concurrency control but establishes how fast a key-value store can be;
    and (2) FUSEE~\cite{fusee}, the state-of-the-art for disaggregated key-value stores, which is replicated and supports concurrency.
Through YCSB benchmarks, we observe that \kvs has sub-RTT
    latency overhead compared to the raw baseline (0.5\,\us for \getop{}s and 1.5\,\us for \updateop{}s), and has
    significantly lower latency than FUSEE
    (up to 2$\times$ faster \getop{}s and 3.4$\times$ faster \updateop{}s),
    as well as better availability (no downtime).
The cost of using \kvs is higher disaggregated-memory consumption (2$\times$ FUSEE's).

In summary, our contributions are the following:
\begin{myitemize}
    \item \sysname: the first replication protocol for disaggregated
      memory
      to offer
      strong consistency
        (linearizability), strong liveness (wait-freedom), and \readop{}s and \writeop{}s that complete in one \rt
        in the common case. \sysname need not run
        code at the memory nodes.
    \item \protocol: a new protocol that directly provides
         the above features if
         memory nodes can 
         provide atomic conditional updates to large objects in one \rt.
    \item \innout: a novel
        technique for the atomic and conditional update of 
        large disaggregated memory objects in one \rt
        with no compute at memory nodes.
    \item \kvs:
        a low-latency, strongly consistent
        and highly available disaggregated key-value store.
    \item A thorough evaluation of \kvs using YCSB,
        which finds much improved latency and availability against the state of the art.
\end{myitemize}

\kvs is open-source, available 
  at \url{https://github.com/LPD-EPFL/swarm-kv}.
The appendix has detailed correctness
  proofs and \rt complexity analyses.
\section{Background} \label{sec:background}

\subsection{Setting} \label{sec:background:dm} \label{sec:setting}

We consider a data center system where servers host applications that can access disaggregated memory.
This disaggregated memory is provided by a number of \emph{memory nodes}
  connected to
  a low-latency high-bandwidth interconnect.
Our goal is to replicate
  the data in disaggregated memory
  across a subset of memory nodes,
  through a protocol
  that application threads use
  to read and write data.
Our protocol, \sysname,
is agnostic to the disaggregation technology, as long as it satisfies
  three properties.
First,
the memory supports \readop{} and \writeop{} operations, though these
  need not be atomic (e.g., concurrent {\writeop}s may clobber
  each other, while a \readop concurrent with a \writeop may return partly written data).
Second,
  the memory supports a 64-bit atomic \casop (compare-and-swap).
Third,
  application threads can pipeline
  two update
  operations
  to be executed in order
  at the same memory node (i.e., if the second update is visible, so is the first), 
  and they execute in one \rt.
This setting can be realized with RDMA~\cite{pfister2001infiniband,beck2011roce-performance},
  but we believe it could also be realized with other technologies like CXL~\cite{intelcxl} in the future.
Although \sysname makes no synchrony assumption for replication, it requires eventual synchrony to recycle memory.

\subsection{Failure Assumptions}

The system is subject to crash failures that affect
  application threads and memory nodes. 
Any number of application threads can fail, but
  we assume a majority of memory nodes
  remains alive (e.g., a deployment
  with 3 or 5 memory nodes can tolerate 1 or 2
  failed nodes, respectively).
We do not consider Byzantine failures.
The network and disaggregated memory
  interconnect may also fail,
  but they eventually recover.
The system remains safe under 
  partitions, but application threads need to be able
  to reach a majority
  of memory nodes to make
  progress.
While the set of application threads can change over time, \sysname operations are
wait-free only if the number of concurrent application threads is bounded.
  
\subsection{Read-Write Replication and the ABD Protocol}
\label{sec:background:ar}
\label{sec:background:abd}

We are interested in strongly consistent replication providing
  linearizability~\cite{herlihy1990-linearizability}, which ensures that operations take effect
  atomically at a single point in time.
Linearizable replication schemes can be divided into two types.
State machine replication~\cite{lamport-clocks, lamport-smr}
  implements state machines with arbitrary operations, while 
  read-write replication~\cite{abd-emulation} implements a \emph{register} object
  with a \readop and a \writeop operation, such that
  \readop{}s return the value of the latest \writeop{}~\cite{attiyawelch}.
State machine replication tends to be more complex and less efficient than read-write
  replication because it requires
  solving consensus~\cite{lamport-paxos}
  which incurs more \rts (e.g., clients first contact a leader, then the leader runs a consensus protocol), and is not wait-free (it can be delayed for arbitrarily long in periods without synchrony).
Therefore, we focus on read-write replication, which is simpler and
  suffices for many use cases, such as key-value stores
  and storage systems~\cite{dutta2004howfast, burkhard2009efficiency, huang2020finegrained, tseng2023distributed}.

\StartLineAt{1}

\begin{lstlisting}[caption={ABD expressed using a max register},label={alg:abd}]
M = (ts: (i: 0, tid: @$\bot$@), v: @$\bot$@) // Max Register

def WRITE(v): // this block is not atomic
  fresh_ts = (M.READ().ts.i + 1, tid)
  M.WRITE((fresh_ts, v))

def READ():
  return M.READ().v
\end{lstlisting}

The canonical read-write replication protocol is ABD~\cite{abd-emulation,abd-mwmr},
  named after its creators Attiya, Bar-Noy, and Dolev.
Algorithm~\ref{alg:abd} shows ABD, expressed
 using a \emph{max register}.
A max register is an
  object that supports \readop{} and \writeop{} operations
  such that a
  \readop{} returns a value
  greater than or equal to the values of all
  operations that completed before it
  (see Appendix~\ref{app:wmr} for more details).\footnote{Our definition of a max register is weaker 
  than~\cite{maxregister}, but suffices for us.}
  
In ABD,
  each value written to the max register is augmented with a logical timestamp used to order it.
This timestamp
  comprises a
  thread id
  to break ties of otherwise
  identical timestamps from different writers.
To \writeop{}, a thread first picks a \emph{fresh} timestamp---one
  that is higher than the timestamps of all
  completed \writeop{}s---by
  reading the max register and adding 1 to the timestamp it finds.
Then, it writes its value alongside the fresh timestamp to the max register.
  To \readop{}, a thread reads the max register and
  returns the value it finds, ignoring the timestamp.
  
The max register itself is implemented from multiple
  max registers that may crash, assuming a majority remains alive. Briefly, to \writeop{} to the max register, we write to a majority of the underlying
 registers.
To \readop{}, we read a majority of the underlying registers,
  pick the largest value, ensure it is written to a majority
  of the underlying registers (which may require an additional \rt in unlucky cases), and return it. 
Appendix~\ref{app:wmr} provides the pseudocode
and a full proof of correctness of this reliable max register implementation.

\subsection{Challenges}

There are many challenges in
achieving our goal of replicating shared data in disaggregated memory with strong
  consistency, strong liveness, and low latency.

\paragraph{One-\rt.}
While ABD offers strong consistency and liveness, its \writeop{} operation
  incurs two \rts to disaggregated
  memory, which doubles the latency compared to a non-replicated system.
An ideal replication scheme would always take one \rt, but
  this is provably
  unachievable~\cite{dutta2004howfast}, so
  we settle for one \rt most of the time.
The ABD algorithm suggests that there is little opportunity to save a \rt: because writing the value to disaggregated memory is mandatory, we can eliminate only the first \rt, used to obtain a fresh timestamp.
We can eliminate this step by guessing a fresh timestamp, but doing so is tricky:
if the guessed timestamp is not fresh, the \writeop{} must be retried, which
  effectively writes the value with two different timestamps; this can cause
  \readop{}s to oscillate between two values, violating linearizability.
We address this challenge with \protocol, a novel replication protocol (Section~\ref{sec:replication}).

\paragraph{The limits of disaggregated memory.}
ABD is designed to replicate data in the memory of processes that can run arbitrarily logic.
Disaggregated memory, however, is limited in its compute capabilities, which makes it hard
  to implement ABD's max register.
The latter indeed
  requires a conditional update (i.e., only write if the
  provided value is larger than the current value) and atomic \readop{}s/\writeop{}s on large values
  (larger than the word size).
While one can extend the atomicity of disaggregated memory by 
  leveraging 64-bit \casop{}es,
  these incur high latency because of locking or pointer chasing~\cite{tsai2020clover}.
We address this challenge with \innout, a novel technique to conditionally and
  atomically manipulate large disaggregated memory buffers in one \rt
  (Section~\ref{sec:ino}).

\section{\protocol} \label{sec:replication}

\protocol is \sysname's core replication protocol.
Section~\ref{sec:replication:overview} gives an informal overview.
 Section~\ref{sec:replication:algorithms} gives its algorithms.
  Section~\ref{sec:replication:splitter} explains how to
  implement \protocol' new building block, the \emph{timestamp lock}.
A full proof of correctness is given in Appendices~\ref{app:timestamp} (timestamp lock) and \ref{app:safeguess} (\protocol).
  
\subsection{Overview} \label{sec:replication:overview}

\begin{figure}
  \centering
  \includegraphics[width=\columnwidth]{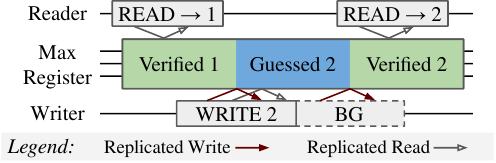}
  \caption{ Fast \readop{}s complete in one \rt to the replicas in disaggregated memory
  by finding a value tagged as verified.
  Fast \writeop{}s complete in one \rt by 
  writing their value with a guessed timestamp and
  confirming the freshness of the latter via a parallel \readop{}.
  Successful \writeop{}s are tagged as verified in the background.}
  \label{fig:sg-fast-path}
\end{figure}

  \protocol adopts a fast-slow path design
so that
\readop{}s and \writeop{}s complete in one \rt in the common case,
  when there are no failures or contention, and clocks are
  nearly synchronized
  (Figure~\ref{fig:sg-fast-path}).
  In other cases, \protocol runs a slow path with more \rts, but
  preserves
  wait-freedom.
  
Similarly to the ABD protocol (\S\ref{sec:background:abd}),
  \protocol orders \writeop{}s using timestamps.
A timestamp $t$ associated with a \writeop{} $W$ is \textit{fresh} if $t$ is greater than the timestamps of all \writeop{}s that completed before $W$ started; otherwise, $t$ is \textit{stale}.
ABD requires obtaining a fresh timestamp for each \writeop{}, which incurs an additional network \rt.
 \protocol avoids this overhead by guessing a timestamp
  that is fresh most of the time, but not always.
When the timestamp is stale, \protocol retries
  with a new fresh timestamp.
 
  While using stale timestamps might seem harmless,
merely requiring a retry,
  a major difficulty lies in diagnosing such
  timestamps as stale.
Sometimes, a thread cannot tell whether
  the timestamp it used was fresh.
To preserve safety in spite of stale timestamps,
  \protocol' slow path retries \writeop{}s with a different
  timestamp only after ensuring that no thread has/will
  ever read the value using the guessed timestamp.

Each \protocol-replicated object
 is built
  on top of two types of reliable wait-free objects:
  a \emph{max register} as in ABD,
  and a set of \textit{timestamp locks}.
The max register stores
  the latest written value alongside its timestamp
  and a flag that indicates whether this timestamp was verified or guessed.
The timestamp lock (one per writer) helps threads
  determine if 
  a value with
  a guessed timestamp
  can be 
  read or rewritten.
  
\subsection{Read and Write Algorithms} \label{sec:replication:algorithms}

Fast \protocol \writeop{}s require writers
  to efficiently guess fresh timestamps in the common case.
  Guessing is encapsulated in the \t{guessTs} function (line~\ref{alg:sg-write:guess-ts}), which returns a pair, the
  second element of which is the writer's thread id.
\protocol is oblivious to \t{guessTs}' implementation but
  mandates it to be strictly
  monotonic at a given thread.
Assuming reasonable clock synchrony among the application threads, a good timestamp can be guessed using a core's local clock.

\StartLineAt{1}

\begin{lstlisting}[caption={\protocol{}' $\writeop{}$ Pseudocode},label={alg:sg-write}]
M = ((0, @$\bot$@), VERIFIED, @$\bot$@) // Max Register@\label{alg:sg-write:M}@
TSL[tid] = {} // Timestamp Lock @\label{alg:sg-write:S}@

def WRITE(v):
  w = (guessTs(), GUESSED, v) @\label{alg:sg-write:guess-ts}@
  i@$ $@n parallel {m = M.READ(), M.WRITE(w)} @\label{alg:sg-write:parallel}@
  if m <= w: // Fast path (fresh timestamp) @\label{alg:sg-write:fp}@
    i@$ $@n bg: M.WRITE(w with VERIFIED) // Spdup reads @\label{alg:sg-write:bg-write}@
  else: // Slow path (potentially stale timestamp) @\label{alg:sg-write:sp}@
    if TSL[tid].TRYLOCK(w.ts, WRITE): @\label{alg:sg-write:split}@
      M.WRITE(((m.ts.i+1, tid), VERIFIED, v)) @\label{alg:sg-write:write}@@\label{alg:sg-write:verified}@
\end{lstlisting}

Algorithm~\ref{alg:sg-write} shows the logic of \protocol' replicated \writeop{}s.
The writer synchronizes with the other threads
  via two reliable shared objects:
  a max register \t{M} (line~\ref{alg:sg-write:M}) identical to ABD's (\S\ref{sec:background:abd}),
  and a timestamp lock \t{TSL[tid]} (line~\ref{alg:sg-write:S}) whose details we explain later.
  The writer starts by preparing the
  3-tuple it will
  write to \t{M}
  (line~\ref{alg:sg-write:guess-ts}).
This 3-tuple is composed of
  (1) the timestamp returned by \t{guessTs}
  thanks to which the writer will try to 
  overwrite the max register (tuples are ordered lexicographically) and 
  linearize its \writeop{},
  (2) a \t{GUESSED} flag that indicates that the writer is not sure
  about the freshness of the timestamp,
  and (3) the value \t{v} it wants to write.
Then, in parallel, the writer reads \t{M}
  while writing its 3-tuple to it
  (line~\ref{alg:sg-write:parallel}).
Although reading \t{M} while writing to it
  (without ordering guarantees) might seem futile,
  it actually allows the writer to further
  establish whether the timestamp it guessed
  was fresh or not.
More precisely, if the 3-tuple it reads is
  less than or equal to the one it wrote
  (line~\ref{alg:sg-write:fp}), then it
  knows that the timestamp it guessed was fresh
  and that its \writeop{} was successfully linearized.
In such cases, \t{WRITE} returns immediately after
  having scheduled a background task to set
  the flag in the 3-tuple to \t{VERIFIED} (which is greater than \t{GUESSED} with respect to the ordering function used by the max register)
  to let readers know
  that its associated timestamp is fresh and definitive
  (line~\ref{alg:sg-write:bg-write}).
This scenario lets \t{WRITE} complete
  in a single access latency to \t{M} in the common case.

Otherwise, i.e., if the writer reads a tuple
  greater than the one it was trying to
  write (line~\ref{alg:sg-write:sp}),
  it cannot assess the freshness of the timestamp it used:
  it may have been stale, but it may also
  have been fresh but overwritten before the 
  writer could read it.
  In the first case, the \writeop{} should be
  re-executed with a fresh timestamp.
In the second, the written value
  could have been read by another thread
  and re-executing the \writeop{} with another
  timestamp would result in \t{v} being written twice,
  breaking consistency.

The writer can safely
rewrite its value with another timestamp using a \textit{timestamp lock}, a novel wait-free mechanism. Briefly,
  this lock lets the writer stop concurrent readers from observing the guessed timestamp. Symmetrically, readers can use it to prevent a rewrite if they deem the guess fresh.
  Similarly to a
  readers-writer lock,
  a timestamp lock
  can either be held
  by multiple readers,
  or by the writer, but never both.\footnote{Our
  formal specification allows
  multiple writers to hold the same lock, but we
  ensure writer exclusion by having one lock per
  writer.}
  Differently,
  (1) timestamp locks are never unlocked 
  (but can be relocked at higher timestamps);
  (2) 
  a reader and a writer trying to lock concurrently may both fail; and (3) locking may also fail if a
  higher timestamp was used.
We formally define and construct timestamp locks in Section~\ref{sec:replication:splitter}. 

\protocol uses timestamp locks as follows. A writer whose guessed timestamp may have been stale tries 
  to
  lock concurrent readers out
  (line~\ref{alg:sg-write:split}). If it fails to, this means that
  a reader also tried to lock the guessed timestamp, which implies that this reader deemed the timestamp fresh, so \t{WRITE} can safely return.
  If the writer manages to lock readers out,
  none
  could have returned the value at the guessed timestamp,
  and any reader
  attempting to do so in the future
  would not be able to lock the timestamp, causing it to retry its \readop{}.
  The writer can thus safely re-execute
its \writeop{} with a 
  timestamp higher than \t{m.ts}---so that it is
  fresh---and
  mark
  it as \t{VERIFIED} (line \ref{alg:sg-write:write}).
\t{WRITE} returns once the \writeop{} to \t{M} is over.
As \t{WRITE} has a unidirectional control flow
  and uses only wait-free objects, it is wait-free.

Algorithm~\ref{alg:sg-read} shows the logic of \protocol' \readop{}s.
Readers use the same max register
  \t{M} and timestamp locks \t{TSL[*]} as writers. Informally, \protocol' \readop{}s consist
  in iteratively reading tuples from \t{M}
  (lines~\ref{alg:sg-read:loop}--\ref{alg:sg-read:read})
  until one is deemed valid and its value is returned.
If a tuple read from \t{M} is marked
  as \t{VERIFIED}, its value is immediately
  returned (line~\ref{alg:sg-read:fp}).
This lets \readop{}s usually complete
  in a single access latency to \t{M}.

Readers, however, cannot wait to find
  a \t{VERIFIED} tuple
  as this would violate both fault tolerance
  and wait-freedom.
They must instead identify, out of all
  the \t{GUESSED} tuples they read from \t{M},
  one with a value that is
  is safe to return.
The condition for a tuple read from \t{M} to be
  safe is two-fold.
First, its timestamp must
  have been fresh (i.e., the writer must
  have used a timestamp higher than all
  previously completed \writeop{}s).
Second, its writer should not re-execute
  its \writeop{} with another timestamp.

\ifacm\else
\ContinueLineNumber

\begin{lstlisting}[caption={\protocol{}' \readop Pseudocode},label={alg:sg-read}]
// Shares M and TSL with writers @\label{alg:sg-read:M}@@\label{alg:sg-read:S}@

def READ():
  seen: dict<ThreadId, MValue> = {} @\label{alg:sg-read:dict}@
  while True: @\label{alg:sg-read:loop}@
    m = M.READ() @\label{alg:sg-read:read}@
    if m is VERIFIED: return m.v // Fast path @\label{alg:sg-read:fp}@
    if m in seen.values: // Fresh timestamp @\label{alg:sg-read:fresh}@
      if TSL[m.ts.tid].TRYLOCK(m.ts, READ):@\label{alg:sg-read:split}@
        i@$ $@n bg: M.WRITE(m with VERIFIED) // Spdup rds @\label{alg:sg-read:write}@
        return m.v @\label{alg:sg-read:return-1}@
    elif m.ts.tid in seen.keys: // Wait-free path@\label{alg:sg-read:in-seen}@
      return seen[m.ts.tid].v @\label{alg:sg-read:return-2}@
    seen[m.ts.tid] = m @\label{alg:sg-read:update-seen}@
\end{lstlisting}
\fi

\begin{figure}
  \centering
  \includegraphics[width=\columnwidth]{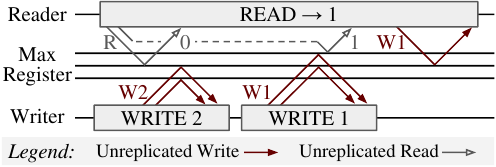}
  \caption{A max register \readop{} reports the wrong maximum out of \writeop{}s concurrent to it.
  Mismatching majorities lead to a \readop{} of 1, despite 1 being written after 2.
  However, no subsequent \readop{} can return 1, as \writeop{} 2 will be over.
  }
  \label{fig:weak-max-register}
\end{figure}

A reader is sure of the freshness of a tuple's timestamp
  after reading it twice from the max register.
Reading it once is not enough, because a \readop{} to a max register can report the wrong maximum out of \writeop{}s concurrent to it
  (Figure~\ref{fig:weak-max-register}).
As the second \readop{} is sequential to the first, it disallows such behavior and confirms the freshness of the timestamp.
  Concretely,
  the reader tracks the
  tuples it read in a dictionary (line~\ref{alg:sg-read:dict})
  that it updates at the end of each iteration
  (line~\ref{alg:sg-read:update-seen}),
  and only considers a tuple as having had a fresh
  timestamp if it has seen it in a previous iteration
  (line~\ref{alg:sg-read:fresh}).

After finding a tuple with a fresh timestamp,
  the reader ensures that its writer
  did not and will never retry its \writeop{} with another
  timestamp by 
  trying to lock the guessed timestamp
  in read mode (line~\ref{alg:sg-read:split}).
  If this succeeds, the reader
  deems the tuple valid.
  The reader thus marks the timestamp as \t{VERIFIED} in the background
  to speed up future \readop{}s (line~\ref{alg:sg-read:write}),
  and returns its associated value (line~\ref{alg:sg-read:return-1}).

In case a reader is never able to lock any timestamp, it can still return after seeing two different
  tuples from the same writer (lines~\ref{alg:sg-read:in-seen}--\ref{alg:sg-read:return-2}).
Indeed, since this writer
  started a second \writeop{},
  the first must have completed (either during the \readop{} or just
  before it),
  so its value is safe to return.
  
Thanks to the assumption of a bounded number of writers, this algorithm
  is forced to terminate in a bounded number of steps.
The key observation is that,
  if a reader fails to 
  lock a timestamp in read mode,
  then its writer must have seen
  an even higher timestamp,
  which ensures that the reader will discover
  a new tuple in the next iteration.
Given that a reader can only loop as long as it does
  not see two different tuples from the same writer
  and that there is a bounded number of writers,
  the reader will return a value in at most
  $2\times\#writers+1$ iterations.
Moreover, as all accessed objects are wait-free,
  so is \t{READ}.

\subsection{Building Timestamp Locks} \label{sec:replication:splitter}

 Formally, a timestamp lock has a single \t{TRYLOCK(ts, READ or WRITE)$\rightarrow$bool} method.
 This method has two properties:
(1) \t{TRYLOCK(ts,m)} returns \t{True} if
  there are no preceding or concurrent calls
  to \t{TRYLOCK(ts,$\neg$m)} (with $\neg \t{READ}{=}\t{WRITE}$)
  or to \t{TRYLOCK(ts',$\star$)} with \t{ts'}>\t{ts},
and (2) it is impossible for \t{TRYLOCK(ts,READ)} and \t{TRYLOCK(ts,WRITE)} to both return \t{True}.
Timestamp locks are similar to conflict detectors~\cite{aspens2014tight}
  and splitters~\cite{moir1995splitter}, but are more space efficient since a single timestamp lock
  handles many timestamps.
  
  A wait-free reliable timestamp lock can be built
  using a set of fallible CAS (compare-and-swap) objects,
  the majority of which is assumed not to fail (we will deploy a
  fallible CAS object on each memory node). These CAS objects
  provide a single
  \t{CAS(x,y)} 
  method,
  which atomically sets the object's
  value to \t{y} if it was previously
  equal to \t{x}, and returns the
  previously stored value.

Intuitively, both lock modes
  race to take control of a majority of 
  the CAS objects, and
  fail if they hear about
  the other in the process.
The key observation here is that
  it is impossible for both lock modes
  to be written at a majority,
  which ensures that at most one side
  can return \t{True}.
  
\StartLineAt{1}

\begin{lstlisting}[float, caption={Timestamp Lock Pseudocode},label={alg:split}]
CASes = {(@$\bot$@, @$\bot$@), ...} // 2f+1 CAS Objects @\label{alg:split:CASs}@

def TRYLOCK(ts, mode: READ or WRITE):
  read: dict<CAS, CasValue> = {(@$\bot$@, @$\bot$@), ...}
  async for c in CASes: @\label{alg:split:async}@
    while read[c].0 < ts: @\label{alg:split:while}@
      expected = read[c]
      read[c] = c.CAS(expected, (ts, mode))
      if read[c] == expected: break@\label{alg:split:cas}@
  wai@$ $@t fo@$ $@r a majority to complete @\label{alg:split:wait}@
  if any c st read[c].0 > ts:@\label{alg:split:any-greater}@ return False @\label{alg:split:return-1}@
  if any c st read[c] == (ts, @$\neg$@mode):@\label{alg:split:any-other-side}@ return False@\label{alg:split:return-2}@
  return True @\label{alg:split:return-3}@
\end{lstlisting} 

Algorithm~\ref{alg:split} details
  the pseudocode of \t{TRYLOCK}.
The locker
  first CASes the tuple
  \t{(ts,mode)} to each individual CAS
  object until a majority of them
  is associated to a timestamp greater than or
  equal to \t{ts} (lines~\ref{alg:split:async}--\ref{alg:split:wait}).
Line~\ref{alg:split:while} ensures that
  no CAS object ever changes its
  opinion on the lock mode
  of a given timestamp, which ensures
  the safety of the timestamp lock.
If any of the CAS objects
  contains a timestamp higher than
  \t{ts}, \t{TRYLOCK} returns \t{False}
  as \t{TRYLOCK(ts',$\star$)} was called with
  \t{ts'{$>$}ts}
  (line~\ref{alg:split:any-greater}). \t{TRYLOCK} also returns \t{False} if any
  of the CAS objects contains \t{ts}
  with the opposite lock mode (line~\ref{alg:split:any-other-side}). Otherwise, \t{TRYLOCK} returns \t{True} (line~\ref{alg:split:return-3}).
\section{\innout} \label{sec:ino}

Section~\ref{sec:background:abd} explains how
  to build
  ABD's and \protocol' max registers
  using a set of failure-prone max registers.
  Implementing these primitive max registers
  over message passing with compute-capable replicas
  is simple, but doing so efficiently 
  over disaggregated memory is challenging.
The difficulty lies in its limited compute capability and
  lack of atomic \readop{}s and \writeop{}s of large values, that is,
  values larger than the word size
  (8\,B in RDMA~\cite{rdma-manual} and CXLv3~\cite{cxl-atomicity}):
  such operations can co-mingle their execution if executed concurrently.
For example, a \readop{} may return only half the data of a concurrent \writeop.
  This section explains how to overcome these
  issues using a new technique we call \emph{\innout}.

\subsection{Standard Approaches}

There are two standard ways of implementing atomic
  \readop{}s and \writeop{}s of large values.
  
\paragraph{In-place updates with locks.}
The first technique consists in
  employing readers–writer locks
  to ensure that no writer modifies
  the data while it is being read.
This technique has three drawbacks:
(1) it does not tolerate thread
  failures, (2) it requires additional
  \rts to take the lock,
and (3) it is not wait-free.

\paragraph{Out-of-place updates with atomic pointer swing.}
The second technique consists in
  writing new data
  to a separate buffer,
  and then atomically
  updating an 8\,B pointer to this buffer.
This technique is wait-free but
  has two performance drawbacks:
(1) \writeop{}s must completely
  write the buffer before updating the pointer (which, depending on ordering guarantees, might require an additional \rt),
  and (2) \readop{}s require an extra \rt
  to chase the pointer.

\subsection{Overview}

\innout combines both in-place and out-of-place updates.
It uses
in-place updates without any locks
  so that \readop{}s can complete in one \rt
  (no need to chase pointers).
However, \readop{}s may sometimes find garbled data.
In those cases, \innout relies on out-of-place updates
  so that these \readop{}s can return correct data.
Importantly, in-place data is validated via a hash that includes a pointer to out-of-place data, so that the former is deemed valid only if it matches the latter.
To update the out-of-place buffer and its pointer
  in a single \rt, \innout leverages the ordering
  guarantees of disaggregated memory (\S\ref{sec:background}).

\subsection{Read and Write Algorithms}

\begin{figure}
  \centering
  \includegraphics[width=\columnwidth]{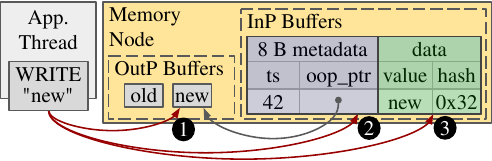}
  \caption{Outline
  of an \innout \writeop.
  In one \rt, the thread
  (1) writes to an out-of-place buffer,
  (2) updates the metadata to point to it,
  and (3) updates the in-place data.
  }
  \label{fig:innout-write}
\end{figure}

Figure~\ref{fig:innout-write} shows 
  \innout's memory layout for
  a given max register.
This layout is divided in three segments located
  on the same memory node:
(1) private memory buffers used
  to safely store data out of place,
(2) 8\,B of metadata that
  stores the data's timestamp and
  a pointer to its out-of-place buffer,
and (3) a buffer for in-place data
  augmented with a hash.
For efficiency, the in-place data
  buffer is located next to the
  8\,B metadata so that both can
  be read in a single operation.

\StartLineAt{1}

\begin{lstlisting}[caption={\innout{}'s Max \writeop{} Pseudocode},label={alg:ino-write}]
def WRITE(v): @\label{alg:ino-write:signature}@
  ts = timestamp(v) @\label{alg:ino-write:timestamp}@
  oop_ptr = out_of_place_buffers.allocate() @\label{alg:ino-write:allocate}@
  pipelined_disaggregated_memory_operations([@\label{alg:ino-write:rdma-seq-0}@
    WRITE(oop_ptr, v),@\label{alg:ino-write:rdma-seq-1}@
    MAX(&metadata, (ts, oop_ptr))])@\label{alg:ino-write:rdma-seq-2}@
  i@$ $@n bg: WRITE(&data, (v, hash((ts, oop_ptr), v))@\label{alg:ino-write:rdma-seq-3}@
  wait for MAX to complete @\label{alg:ino-write:wait}@
\end{lstlisting}

Algorithm~\ref{alg:ino-write} shows
  the pseudocode for \innout \writeop{}s.
As it specifically implements ABD's and \protocol' large
  \emph{max} registers,
  it assumes that the written values 
  include a small timestamp that
  orders them (\protocol' \t{GUESSED} or \t{VERIFIED} flag is encoded in its least significant bit).
A \writeop starts by extracting
  said timestamp using a
  \t{timestamp} function (line~\ref{alg:ino-write:timestamp}).
Then, the writer allocates a
  new out-of-place buffer (line~\ref{alg:ino-write:allocate}).
Writers pre-allocate large memory chunks
  so that this buffer allocation does not require
  accessing disaggregated memory.
  The writer then issues a
  series of two memory operations
  (lines~\ref{alg:ino-write:rdma-seq-0}--\ref{alg:ino-write:rdma-seq-2})
  that get executed in order thanks
  to hardware support for FIFO pipelining (\S\ref{sec:background:dm}).
The first operation fills the just-allocated
  out-of-place buffer with the
  value to be written (line~\ref{alg:ino-write:rdma-seq-1}).
The second operation atomically
  sets the 8\,B metadata to a pair made of
  the timestamp and the out-of-place pointer
  if this is higher than the tuple it
  previously stored (line~\ref{alg:ino-write:rdma-seq-2}).
The writer then schedules a background update of the in-place buffer with the written value and its hash (line~\ref{alg:ino-write:rdma-seq-3}).
The \writeop is deemed over when
  the \maxop operation completes (line~\ref{alg:ino-write:wait}).

\ContinueLineNumber

\begin{lstlisting}[float,caption={\innout{}'s \readop{} Pseudocode},label={alg:ino-read}]
def READ():
  (metadata, data) = READ(&metadata_data) @\label{alg:ino-read:read}@
  if hash(metadata, data.value) == data.hash:@\label{alg:ino-read:crc}@
    return data.value @\label{alg:ino-read:return}@
  return READ(metadata.oop_ptr) // fall back @\label{alg:ino-read:fallback}@
\end{lstlisting}

Algorithm~\ref{alg:ino-read} shows the pseudocode for \readop{}s.
The thread first reads
  the 8\,B metadata and the
  in-place data (line~\ref{alg:ino-read:read}).
If the hash of the in-place data is correct (i.e., the in-place data matches the out-of-place data),
  its value is returned
  (lines~\ref{alg:ino-read:crc}--\ref{alg:ino-read:return}).
Otherwise, the value is read from the out-of-place
  buffer using the pointer in
  the metadata (line~\ref{alg:ino-read:fallback}).

\subsection{The MAX Operation} \label{sec:ino:max}

\innout's \t{WRITE} procedure (Algorithm~\ref{alg:ino-write})
  uses a \maxop{}
  read-modify-write atomic operation.
While this operation is conceptually similar to
  a \casop (compare-and-swap),
  it is typically not supported in hardware.
We thus emulate this operation
  using a \casop (Algorithm~\ref{alg:max}).

\StartLineAt{1}

\begin{lstlisting}[caption={\casop{}-Based \maxop{} Replacement},label={alg:max}]
def MAX(dst, v):
  prev = @$\bot$@
  while prev < v:
    expected = prev
    prev = CAS(dst, expected, v)
    if prev == expected: break
\end{lstlisting}

This emulation has two drawbacks.
First, while a native \maxop operation would always
  complete in one \rt, this implementation
  can retry a (bounded) number of times.
Second, this replacement needs a \rt to discover
  the previous value of the 8\,B buffer.
In practice, as \writeop{}s often follow \readop{}s,
  the value of the 8\,B buffer
  can be cached for \maxop to complete
  in one \rt.
Still, \writeop{}s to frequently modified keys can suffer
  from contention.

We reduce contention by replacing the 8\,B buffer with an array of
  8\,B buffers, each updated by a subset of the writers.
To find the maximum, readers scan the array and return the
  largest entry.
While the scan is not atomic,
this technique 
  provides the required semantics of a max register: \readop{}s return a value greater than or equal to all operations that completed before them.
Moreover, by using one 8\,B buffer per writer,
  \maxop{}es
  can be
  made to complete in one \rt.
We evaluate the benefits of this technique (\S\ref{sec:eval:micro}).

\subsection{Recycling Memory}\label{sec:ino:rm}

Each \innout \writeop{} allocates
  a new out-of-place buffer (Algorithm~\ref{alg:ino-write}).
These buffers eventually need to be recycled with care,
  lest readers can obtain invalid data due to concurrent (unexpected) \writeop{}s.
One way to recycle is to use wait-free hazard pointers~\cite{wfhp}:
  a reader advertises its wish to read
  data before accessing its out-of-place
  buffer, so that if a writer were to recycle
  its memory, it would first satisfy
  the reader's wish by handing it a valid
  buffer, and only recycle unrequested buffers.
This technique preserves wait-freedom,
  but requires an extra \rt to 
  advertise \readop{}s.
We use a different approach, which does not require the extra \rt:
  to recycle, a thread first asks
  all readers to stop accessing the buffers
  to be freed.
In theory, this approach is not wait-free
  because a writer that runs out of memory must wait for 
  slow readers to respond to its recycling request.
In practice, this is a sensible trade-off as
  real-life systems are partially synchronous
  and use membership services~\cite{hunt2010zookeeper, ukharon},
  so they can
  recycle memory in background tasks with bounded delays,
  so that the read-write protocol never runs out of memory
  and thereby remains wait-free.

\section{\kvs}

We use \sysname
  to build \kvs, the first fault-tolerant disaggregated
  key-value store to provide single-\rt
  \getop{}s, \updateop{}s, \insertop{}s and \deleteop{}s
  in the common case.
  Section~\ref{sec:kvs:setting} establishes the setting, Section~\ref{sec:kvs:overview} gives 
 the outline
of \kvs,
  Section~\ref{sec:kvs:protocols} explains the protocol
  for each operation,
  and Section~\ref{sec:kvs:rm} discusses memory recycling.

\subsection{Setting}\label{sec:kvs:setting}

 SWARM-KV clients run on traditional data-center servers that directly \insertop{}, \updateop{}, \getop{} and \deleteop{} key-value pairs replicated in disaggregated memory using the SWARM-KV library.
There are no intermediary servers between the clients and the data
 in disaggregated memory, which minimizes the latency.
 Any number of SWARM-KV clients can crash, along with any set of memory nodes, provided each key has a majority of its replicas accessible.
 SWARM-KV also requires a fast index as well as a membership service, which can run on traditional servers and are fault-tolerant.
While SWARM-KV is safe under partitions, clients might temporarily lose liveness if they are unable to contact a majority of disaggregated replicas, the index, or the membership service.

\subsection{Overview}\label{sec:kvs:overview}

\begin{figure}
  \centering
  \includegraphics[width=\columnwidth]{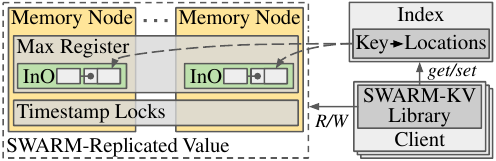}
  \caption{Architecture of \kvs with a single key.
  Clients use 
  \sysname \readop{}s and \writeop{}s to directly access
  the value replicated in disaggregated memory.
  The replicas that form \protocol' max register are spread across different memory nodes and implemented via \innout max registers.
  The location of the replicas is stored in a reliable index.
  }
  \label{fig:swarm-kv-arch}
\end{figure}

Figure~\ref{fig:swarm-kv-arch} shows \kvs's architecture.
  \kvs's clients
  directly access
  key-value pairs replicated over a set of
  memory nodes using \sysname.
Each replica is a max register implemented via \innout.
  An index tracks the locations of keys' replicas (i.e., their memory nodes and addresses on them).
\kvs is oblivious to the choice of index,
  as long as it is reliable and allows clients
  to set and get the replicas
  associated to a key in a
  single \rt in the common case.
When a client accesses a key-value pair for the first time,
  it finds the location of its replicas via the index
  and caches it for subsequent operations.
 This location data is small
  ($\approx$24\,B per key),
  and is the only significant source of memory consumption for clients,
  so clients can
  cache tens of millions of keys.
  
\subsection{Protocols}\label{sec:kvs:protocols}

We now give the protocols for
  {\insertop}~(\S\ref{sec:kvs:protocols:insert}),
  {\deleteop}~(\S\ref{sec:kvs:protocols:delete}),
  {\updateop}~(\S\ref{sec:kvs:protocols:update}),
  and {\getop}~(\S\ref{sec:kvs:protocols:get})
  \kvs operations.

\subsubsection{\insertop{}s.}\label{sec:kvs:protocols:insert}
 The client executes two things in parallel: it replicates the value over
  a set of memory nodes, and
  it inserts the replicas' location in the index.
  More precisely, the client first
 picks a set of memory nodes
  and allocates on each an \innout max-register replica
  (a buffer to fit \innout's 8\,B metadata and in-place data, \S\ref{sec:ino}).
The client then uses a \sysname \writeop{} to replicate
  the inserted value over the chosen replicas.
Importantly, upon the replicas' allocation,
  their \innout 8\,B metadata is cleared
  to indicate that they are empty.
Clients pre-allocate cleared buffers
  so that this step completes in one \rt.

In parallel to this \writeop{}, the client
  asks the index to map the inserted key to the
  chosen replicas.
  If a mapping to replicas marked for deletion already exists (\S\ref{sec:kvs:protocols:delete}), it is overwritten. Otherwise, if a mapping to writable replicas already exists, the client recycles the \innout buffers it just allocated, and the \insertop{} request turns into an \updateop{} (\S\ref{sec:kvs:protocols:update}) that uses the existing replicas.
 The location of the replicas is then
  cached so that future \getop{} and \updateop{} operations on this key
  can bypass the index.
  
The \insertop is over once both tasks complete, which parallelization allows for in one \rt in the common case.
It is fine for the index
  insertion to finish before the \sysname \writeop{},
  as concurrent readers would find
  empty \innout replicas, indicating that the key does not exist.
Similarly, client failure may leave an index entry pointing to empty replicas,
  but this does not affect correctness.

\subsubsection{\deleteop{}s.}\label{sec:kvs:protocols:delete}
The client consults its
cache to find the replicas of a key; if not found, it
  consults the index.
Then, it runs a special \sysname
  \writeop{} that uses the max timestamp
  (i.e., all bits set) to ensure that it
  cannot be overwritten and that future
  \sysname \readop{}s and \writeop{}s will fail until the deleted key is re-inserted.
This way,
  the clients that may have cached
  replicas for the deleted key will see it
  deleted.
  Once the special \sysname
  \writeop is over,
the \deleteop completes and the client simply schedules a background task to unmap the just-written replicas from the key.

\subsubsection{\updateop{}s.}\label{sec:kvs:protocols:update}
The client finds the key's replicas
  using its cache or the index.
Then, it issues a \sysname
  \writeop{} to replicate its value
  and return upon its successful completion.
If the key is not indexed, the \updateop{} fails.
If the replicas were cached
  but the \writeop{} fails due to
  a previous \deleteop{},
  the client flushes its cache,
  deletes
  the old mapping in the index (in case the deleter failed),
  and retries the \updateop{}.
  
\subsubsection{\getop{}s.}\label{sec:kvs:protocols:get}
The client locates the key's replicas
  and then issues a \sysname
  \readop{}.
If the key is not (fully) inserted,
  the \getop fails.
Similarly to \updateop{}s, if the \readop{}
  fails due to a \deleteop{},
  the client flushes its cache
  and retries the \getop{}.

\subsection{Recycling Memory}\label{sec:kvs:rm}

  In the background, clients check if the in-place
  buffers they allocated are still indexed, and if
  their out-of-place buffers are still accessible.
To help with this phase,
  clients who delete or update a key inform its
  allocators that their buffers should be recycled.
Before a client recycles a buffer,
  it ensures other clients no longer
  access it by asking them to remove the
  key from their index cache and to stop accessing
  to-be-recycled out-of-place buffers.
Failed clients may not reply to these
  requests.
We use uKharon~\cite{ukharon}, a fast membership manager,
  to monitor the health of clients and instruct
  memory nodes to disconnect from
  suspected clients 
  so they do not access freed memory.
Some memory nodes may be inaccessible and not do so,
  but, for each key, a majority will be able
  to recycle memory,
  which is enough to ensure
  the liveness of \kvs.
\section{Implementation}\label{sec:implementation}

Our implementation of \kvs
  targets RDMA-based disaggregated memory
  using uKharon's framework~\cite{ukharon},
  and
  consists of 3,151 lines of C++17 (CLOC~\cite{cloc}).
It uses xxHash3~\cite{xxhash} for \innout's hashes
and FUSEE's index~\cite{fusee}
  that we modify to provide strong consistency.
Our clients guess timestamps via a loosely synchronized TSC-based
  clock~\cite{tsc} that they re-synchronize every time they
  guess a stale timestamp.
We did not implement memory recycling (\S\ref{sec:ino:rm} and \S\ref{sec:kvs:rm}).

To optimize bandwidth, memory, and latency,
  our implementation slightly diverges
  from the pseudocode previously given.
First, \protocol optimistically tries to replicates its
  \writeop{}s to a mere majority of the memory nodes
  determined by hashing the keys to spread the load.
If one of the memory nodes does not respond immediately,
  \protocol then tries to replicate to all replicas.
Similarly, \protocol \readop{}s optimistically read from
  a mere majority of the memory nodes.
Second, \innout's in-place data is only stored at one of the
  replicas---also determined by hashing of the key---and
  only when marking a \protocol tuple as verified (\S\ref{sec:replication:algorithms}).
\section{Evaluation} \label{sec:evaluation}

We evaluate \sysname through \kvs with the goal of answering the following questions:

\begin{myitemize}
  \item What is \kvs's latency under YCSB (\S\ref{sec:eval:latency})?
  \item How does throughput impact latency (\S\ref{sec:eval:lat-tput})?
  \item How does \kvs scale with clients (\S\ref{sec:eval:scalability})?
  \item How do different value sizes affect \kvs (\S\ref{sec:eval:value-size})?
  \item What is the impact of the replication factor (\S\ref{sec:eval:factor})?
  \item What is \kvs's resource consumption (\S\ref{sec:eval:res-cons})?
  \item What is the impact of failures on availability (\S\ref{sec:eval:knob})?
  \item What is the impact of extreme contention (\S\ref{sec:eval:contention})?
  \item How does our \maxop substitute handle contention (\S\ref{sec:eval:micro})?
\end{myitemize}

\begin{table}[b]
  \setlength{\tabcolsep}{2pt} \caption{Configuration details of machines.}
	\label{table:hwspecs}
\small
 \centering
	\begin{tabular}{cl} 	\toprule
 \multicolumn{2}{c}{\it 4 Client Servers} \\
\textbf{CPU}		& 2$\times$ 8c/16t Intel Xeon Gold 6244 @ 3.60\,GHz \\
\textbf{NIC}		& Mellanox CX-6 MT28908 \\
\midrule
\multicolumn{2}{c}{\it 4 Memory Nodes} \\
\textbf{Memory}		& 4$\times$ SK Hynix 32GB DDR4-2400 RDIMM @ 2133 MT/s \\
 \textbf{NIC}		& Mellanox CX-4 MT27700 \\
\midrule
\textbf{Switch}		& MSB7700 EDR 100\,Gbps \\
\textbf{Software} & Linux 5.15.0-102-generic / Mlnx OFED 5.3-1.0.0.1 \\
	 \bottomrule
	\end{tabular}
\end{table}

\paragraph{Testbed.}
 Our testbed is a cluster with 4 servers and 4 memory nodes configured per 
  Table~\ref{table:hwspecs}.
The dual-socket machines have an RDMA NIC attached to the first socket.
Our experiments execute on cores of the first socket using local NUMA memory.
We measure time using TSC~\cite{tsc} via \t{clock\_gettime} with the \t{CLOCK\_MONOTONIC} parameter. All machines are connected to a single switch, but we expect
 the latency benefits of \kvs to increase in settings
 with more hops.
 
\paragraph{Baselines.}
We compare \kvs against three baselines.
The first is an unreplicated disaggregated key-value store
  that provides no consistency under concurrent operations, which we call \emph{RAW}.
 Such a system is not useful in practice,
  but it is useful to establish a lower bound on latency.
  
  The second baseline is a disaggregated key-value store replicated via ABD, we call \emph{DM-ABD},
  that supports concurrency via out-of-place updates.
It represents a good engineering solution using
  known techniques. DM-ABD's \getop{}s and \updateop{}s commonly finish in two \rts:
  while \getop{}s require an extra indirection,
  \updateop{}s hide latency by writing out-of-place data in parallel to finding a
  fresh timestamp.

The third baseline is FUSEE, a state-of-the-art disaggregated
  key-value store that uses a synchronous replication scheme. FUSEE uses out-of-place updates that take at least four \rts.
FUSEE caches \updateop{}s' location
  so that \getop{}s run in one \rt if they
  re-access a key before it is updated, but take
  two \rts otherwise. 
  Due to FUSEE's stronger synchrony assumptions and higher latency, we limit the comparison with it to end-to-end latency (\S\ref{sec:eval:latency}).
  
 \begin{table}[]
\caption{Number of \rts for \getop{}s and \updateop{}s.}
\begin{tabular}{lccccc}
  & \multicolumn{2}{c}{Common} & \multicolumn{2}{c}{99th percentile} \\ & \getop & \updateop & \getop & \updateop \\ \hline RAW & 1 & 1 & 1 & 1 \\ \kvs & 1 & 1 & 1 & 1 \\ DM-ABD & 2 & 2 & 2 & 2 \\ FUSEE & 1--2 & 4 & \textcolor{purple}{\textbf{2}} & \textcolor{purple}{\textbf{5}} \\ \end{tabular}
\label{table:rtts}
\end{table}

Table~\ref{table:rtts} 
  shows
the number of \rts for
  \getop{}s and \updateop{}s of all systems, for
  both the common case and the 99th percentile, as observed under a standard workload (\S\ref{sec:eval:latency}).
  In all systems, operations are so fast that same-key contention is rare.
Due to their use of a synchronous index, all 
 systems have a theoretical unbounded latency in the worst case~\cite{fischer1985flp}.

\paragraph{Workloads.}
We run YCSB~\cite{ycsb} workloads A
(50\% \getop{}s and 50\% \updateop{}s)
and B (95\% \getop{}s and 5\% \updateop{}s)
with Zipfian (.99) key distribution.
 Unless 
 stated otherwise, we use 3 replicas,
  100K keys of 24\,B, 64\,B values, 4 clients that issue
one operation at a time, and index caches large enough to cache all key locations. There is a warm-up phase of 1M operations and we measure the following 1M operations.

\subsection{Latency} \label{sec:eval:latency}

  We evaluate the latency of all four systems with 4 clients, each issuing 1 operation at a time.
 Figure~\ref{fig:latency-cdf} shows the
  results for YCSB workload B with Zipfian key distribution.

For \getop{}s, RAW has a median latency of 1.9\,\us.
\kvs has a median latency of
  2.4\,\us---a mere 27\% increase over RAW---,
  and a very low spread (P1=2.2\,\us and P99=2.8\,\us).
FUSEE has a bimodal distribution:
  87\% of operations run in one
  \rt in 2.9\,\us (+53\% vs RAW), while keys recently modified by other clients require accessing the index to obtain the new location,
  which takes a second \rt, reaching 4.8\,\us at the 90th percentile (+157\% vs RAW).
DM-ABD comes last with a median latency of
  4.3\,\us (+130\% vs RAW) due to its use of pure out-of-place updates.
 
For \updateop{}s, RAW has a median latency of 1.6\,\us.
\kvs runs in one \rt with a median latency of 3.1\,\us, a 92\% increase over RAW.
DM-ABD always runs in two \rts,
  with a median of 4.9\,\us (+206\% vs RAW).
FUSEE comes last with a bimodal distribution:
  rarely modified keys take 4 \rts in 8.5\,\us (+431\% vs RAW),
  while hot keys take 5 in 10.4\,\us at the 90th percentile (+554\% vs RAW).

Under YCSB workload A (not shown),
  all systems perform similarly, except FUSEE whose
  switch to slower operations begins much earlier at the 41st percentile.
Overall, \kvs adds less than 1 RTT of overhead over RAW,
  and has substantially lower overhead than other replicated
  solutions.

\begin{figure}
  \centering
  \includegraphics[width=\columnwidth]{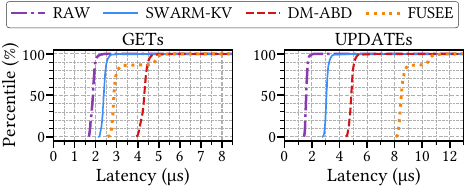}
  \caption{Latency CDFs of RAW, \kvs, DM-ABD and FUSEE with YCSB workload B (95\% \getop{}s and 5\% \updateop{}s), Zipfian key distribution, and 4 clients. }
  \label{fig:latency-cdf}
\end{figure}
\begin{figure}
  \centering
  \includegraphics[width=\columnwidth]{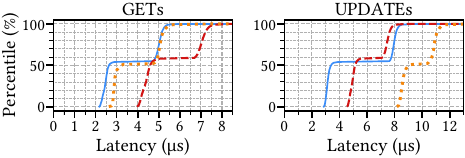}
  \caption{Experiment of Figure~\ref{fig:latency-cdf}, but with 1M keys which do not all fit in index caches, and excluding RAW.}
  \label{fig:limited-cache}
\end{figure}
 
We run a similar experiment for 1M keys with clients using small index caches, and observe how the latency distribution changes across \kvs, DM-ABD and FUSEE (RAW, which merely established a lower bound on latency, is omitted).
We limit the cache size of each
  system to 5\,MiB of system-specific metadata.
More precisely, in DM-ABD and FUSEE,
  cache entries have 24\,B, while
  in \kvs, they have 32\,B since they also include \innout's metadata field (see Figure~\ref{fig:innout-write}).
As a result, DM-ABD and FUSEE can
cache more keys (21.8\% of them),
  compared to \kvs (16.4\% of them).
This accounting does not include
  metadata for the cache replacement
  policy, which is the same
  for all systems (${\approx}$32\,B per entry
  for an approximation of LFU).\footnote{Had we accounted
  for this extra overhead, the
  difference in cache coverage
  between DM-ABD/FUSEE and \kvs
  would be smaller, thus benefitting
  \kvs.
We extend the warm-up to 8M operations to stabilize the cache policy.

Figure \ref{fig:limited-cache}
 shows the results. With this configuration, all systems} exhibit a bimodal distribution. 
On cache hits, the latency distribution follows that of Figure \ref{fig:latency-cdf}. Cache misses add 1 extra \rt to read the index on all systems, except for \updateop{}s on SWARM-KV where 2 additional \rts are required: one to read the index, and one to read the latest metadata buffer.
With such small caches, we measure
 that DM-ABD and FUSEE have a miss rate of 42.5\%, but, similarly
 to Figure \ref{fig:latency-cdf}, 11\% of FUSEE's cache hits
 occur on
stale pointers due to recent modifications by other clients.
 Thus, FUSEE needs to access the index 48.8\% of the time.
While \kvs caches 25\% fewer keys, the combination of a Zipfian distribution with an LFU replacement policy ensures that the remaining 75\% are likely to be the hottest, so \kvs's miss rate only increases to 45.6\%. In this configuration, \kvs's average latency remains better than DM-ABD's and FUSEE's
for both types of operations. Upon a cache miss, \kvs's and DM-ABD's \updateop{} algorithms are equivalent, but the former performs slightly worse as it fetches multiple metadata buffers (\S\ref{sec:ino:max}).
 
\subsection{Per-Core Throughput and Latency under Load} \label{sec:eval:lat-tput}

\begin{figure}
  \centering
  \includegraphics[width=\columnwidth]{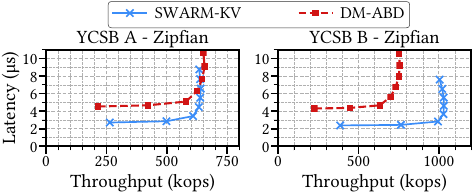}
  \caption{Per-core throughput-latency graph of \kvs and DM-ABD with YCSB workloads A and B, Zipfian key distribution, and 4 clients. Each marker is a measurement for a given number of concurrent
  operations from 1 to 8.}
  \label{fig:eval-lat-tput}
\end{figure}

We study \kvs's per-core
  throughput and latency under load.
Figure~\ref{fig:eval-lat-tput} shows
  the results with each
  single-threaded client
  running
  from 1 to 8
  concurrent operations and DM-ABD as a reference,
  for YCSB workloads A and B with Zipfian key distribution.

Both workloads show large gains in throughput and low impact on latency
  with 3 concurrent operations, before reaching a throughput-latency wall
  as the CPU is bottlenecked by the submission of RDMA requests.
With YCSB A and 1 operation at a time, \kvs's average latency is 2.7\,\us with a throughput of 264\,kops.
Running a second operation has little impact on latency (+5\%), which reaches 2.8\,\us, while nearly doubling the throughput at 499\,kops.
Running a third operation increases the latency by 21\%, reaching 3.4\,\us, but the throughput 
  only increases by 22\% at 609\,kops.
Additional operations increase the latency by 1\,\us each; a maximum throughput of 640\,kops is reached with 6 concurrent operations.
DM-ABD reaches a similar maximum throughput.

YCSB B has fewer \updateop{}s, so \kvs's average latency starts at a lower 2.4\,\us with a throughput of 389\,kops.
Similarly to YCSB A, having more than 3 concurrent operations yields negligible throughput
  enhancement while increasing latency by nearly 1\,\us with each additional operation.
A maximum throughput of 1030\,kops is achieved with 5 concurrent operations.
\kvs outperforms DM-ABD.

  This limited throughput gain
  is due to \kvs's
lack of batching
  and the fixed cost of issuing a
  series of RDMA operations to a memory node.
As the latter takes 200+\,ns and each \getop{}/\updateop{} reaches
 2 memory nodes,
  a core can send 1
  additional operation while waiting for the first to complete, but issuing a second delays the processing of the first.

\subsection{Scalability} \label{sec:eval:scalability}

\begin{figure}
  \centering
  \includegraphics[width=\columnwidth]{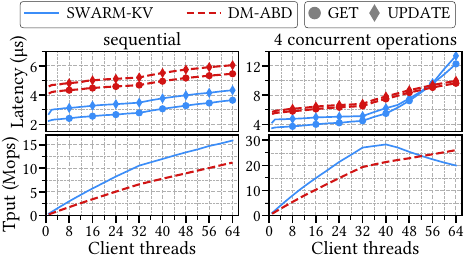}
  \caption{Throughput and average latency of \kvs and DM-ABD with YCSB workload B, 
  Zipfian key distribution, when varying the number of clients and concurrency.}
  \label{fig:eval-scale}
\end{figure}

We evaluate the scalability of \kvs by running the workloads with
  a number of single-threaded clients ranging from 1 to 64 while measuring the latency of \getop{}s and \updateop{}s, and the total throughput.
Figure~\ref{fig:eval-scale} shows the results for YCSB workload B with Zipfian key distribution when clients try to optimize latency (i.e., with a single operation at a time)
  and when clients try to optimize throughput (i.e., with 4 concurrent operations).
We plot DM-ABD for reference.

With 1 operation at a time, \kvs's throughput scales nearly linearly with the number of clients, reaching a maximum of 15.9 Mops with 64 clients.
After 32 clients, the throughput grows slower due to hyper-threading.
The latency of \getop{}s moderately increases from 2.2\,\us with a single client to 3.69\,\us with 64, with virtually all \getop{}s completing in a single \rt.
Similarly, the latency of \updateop{}s, increases from 2.7\,\us to 4.3\,\us with 64 clients.
Both \getop{}s and \updateop{}s have much lower latency than DM-ABD's.

With 4 concurrent operations, the throughput of \kvs scales nearly linearly until 40 clients where it reaches a peak of 28.3 Mops with 5.5\,\us \getop{}s, and 6.22\,\us \updateop{}s.
After this point, \kvs bottlenecks our 100\,Gbps fabric, which leads to a fast increase in latency---underperforming DM-ABD at 56 clients---and a drop in throughput.

\kvs's good scalability comes from its full bypass
  of the CPUs of memory nodes, 
  the contention-handling mechanism of its \maxop substitute (\S\ref{sec:ino:max}), and its efficient guessing of fresh timestamps using loosely synchronized clocks.
 
\subsection{Impact of Value Size} \label{sec:eval:value-size}

We study the effect of different value sizes ranging from 16\,B to 8\,KiB on \kvs's latency and throughput.
Suspecting that large values might perform poorly due to \innout's combination of in-place and out-of-place updates, we run the same experiments with a version of \kvs that does not use in-place updates.
Figure~\ref{fig:eval-size} shows the
  results for YCSB workloads A and B under Zipfian key distribution.

\begin{figure}
  \centering
  \includegraphics[width=\columnwidth]{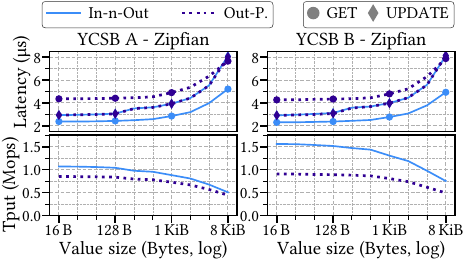}
  \caption{Throughput and average latency of \kvs with YCSB workloads A and B, and varying value sizes. Comparison with a \kvs variant without in-place updates.}
  \label{fig:eval-size}
\end{figure}

Regardless of the workload, the latency of \kvs operations grows linearly with the increase in value size and values
up to 8\,KiB still achieve single-digit latency.

One could expect \innout's combination of in-place and out-of-place updates to
  be detrimental to the latency of large keys: it is not.
Our evaluation shows that, regardless of the YCSB workload and the value size:
  (1) \getop{}s always benefit from in-place data, with large 8\,KiB values still being 33\% faster (by 2.5\,\us),
  (2) \updateop{}s comprising in-place updates are just as fast as purely out-of-place ones thanks to the lazy writing of in-place data (\S\ref{sec:implementation}),
  and (3) \innout's combination leads to higher total throughput, especially for read-intensive workloads (+50\% at 8\,KiB for YCSB workload B).

\subsection{Impact of Replication Factor} \label{sec:eval:factor}

We study the impact of the replication factor
  on \kvs's latency and throughput
  by using 3, 5, or 7 replicas per key.
Due to our setup having only 4 memory nodes,
  we sometimes store 2 replicas on the same node.
Figure~\ref{fig:eval-replication} shows the
  results for YCSB workload B with Zipfian key
  distribution and DM-ABD for reference.
Both systems replicate
  to a mere majority of the replicas (e.g., with 3 replicas, both write to only 2 replicas to tolerate 1 failure) (\S\ref{sec:implementation}).

\begin{figure}[bp]
  \centering
  \includegraphics[width=\columnwidth]{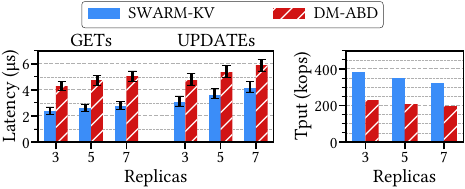}
  \caption{Median latency and per-client throughput of \kvs and DM-ABD with YCSB workload B, Zipfian key distribution and different numbers of replicas. Whiskers depict the 1st and the 99th percentiles.}
  \label{fig:eval-replication}
\end{figure}

Both \kvs and DM-ABD see their latency increase linearly with the number of replicas.
\kvs starts with a median latency of 2.3\,\us for \getop{}s and 3\,\us for \updateop{}s with 3 replicas.
Each 2 additional replicas increase the latency of \getop{}s by 0.2\,\us and \updateop{}s by 0.5\,\us, due to the cost of issuing an additional series of RDMA operations.
  Thus, \kvs's throughput
  drops by 9\% moving from 3 to 5 replicas,
  and by an additional 7\% moving from 5 to 7.
DM-ABD
  shows similar latency increase, 
but starts at 4.3\,\us for \getop{}s and 4.7\,\us for \updateop{}s.
For both systems, the spread in latency between the 1st and the 99th percentiles is stable, increasing by a mere 10\% with each additional 2 replicas.

\subsection{Resource Consumption} \label{sec:eval:res-cons}
Table~\ref{table-res} shows the computed resource consumption of the four systems when tolerating 1 failure, storing 1M keys with values of size 1\,KiB, with 4 clients executing 200\,kops each under YCSB workload B, and running garbage collection (\S\ref{sec:ino:rm} and \S\ref{sec:kvs:rm}) once per second.

\begin{table}
\caption{Resource consumption, for 1M keys, values of 1\,KiB, YCSB workload B, and 4 clients each executing 200 kops.}
\begin{tabular}{lccccc}
  & \multicolumn{3}{c}{\textbf{Per single-threaded client}} & \textbf{Disagg.} \\
  & \textbf{CPU} & \textbf{Cache} & \textbf{IO BW} & \textbf{Mem.} \\ \hline
\textbf{RAW} & 46.6\% & 22.9\,MiB & 6.55\,Gbps & 0.95\,GiB \\ 
\textbf{DM-ABD} & 99.0\% & 22.9\,MiB & 6.99\,Gbps & 3.00\,GiB \\
\textbf{\kvs} & 61.3\% & 30.5\,MiB & 7.41\,Gbps & 4.06\,GiB \\
\textbf{FUSEE} & 74.2\% & 22.9\,MiB & 8.15\,Gbps & 2.04\,GiB \\ \hline
\end{tabular}
\label{table-res}
\end{table}

The cost of \kvs's low latency is a higher use of disaggregated memory. 
RAW, due to its lack of replication, has minimal resource use.
DM-ABD uses 3 replicas per key plus additional metadata, so its consumption of disaggregated memory is 3.1$\times$ higher than RAW's.
However, as DM-ABD \updateop{}s replicate values to only 2 replicas and \getop{}s retrieve the full values from only 1, DM-ABD's bandwidth use is only 7\% higher than RAW's.
 \kvs uses 35\% more disaggregated memory and 6\% more bandwidth than DM-ABD due to \innout and its contention-reduction mechanism (\S\ref{sec:ino:max}).
Moreover,
 \kvs's clients consume 33\% more memory
than DM-ABD's due to their need to cache \innout's 8\,B metadata.
Thanks to its synchronous replication protocol, FUSEE requires only 2 replicas to tolerate 1 failure, which makes it consume 32\% and 50\% less disaggregated memory than DM-ABD and \kvs, respectively.
However, FUSEE's optimistic \getop{}s fail 13\% of the time, wasting bandwidth, so FUSEE consumes 16\% and 10\% more IO than DM-ABD and \kvs, respectively.

Thanks to its low latency, \kvs does little polling, so its CPU use is only 30\%
higher than RAW's, while being 38\% lower than DM-ABD's, and 16\% lower than FUSEE's.

\subsection{Memory Node Failure} \label{sec:eval:knob}
We study the impact of the failure of a memory node
  on \kvs's performance and availability by crashing a memory node while
  measuring \kvs's throughput and latency.
Figure~\ref{fig:eval-failure} shows the
  results.

\begin{figure}
  \centering
  \includegraphics[width=\columnwidth]{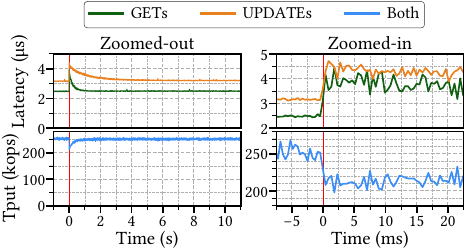}
  \caption{Latency and throughput of a \kvs client with YCSB workload A before and after
  the failure
  (at $t{=}0$)
  of a memory node. Zoomed-out on the left to see performance recovery. Zoomed-in on the right to see the transition.}
  \label{fig:eval-failure}
\end{figure}

The failure of a memory node has no impact on the availability of \kvs, but it temporarily affects its latency.
This impact results from \kvs's bandwidth-reduction mechanism which saves bandwidth in the absence of failure (\S\ref{sec:implementation}).
 Upon a failure, ongoing \kvs operations merely contact additional memory nodes as the majority they initially contacted takes longer than expected to 
respond.
Importantly, this change of replicas requires no reconfiguration of the system, so it has negligible impact.
\kvs's latency increases after a failure due to the loss of in-place data, and the lack of unanimity across the remaining replicas.
Both are eventually rebuilt by subsequent operations, which allows operations 
to recover their latency.

In contrast, synchronous systems like FUSEE reportedly take at least tens of milliseconds to recover from failures~\cite{fusee, hunt2010zookeeper}.
This is because they need to accurately detect failures before initiating recovery procedures involving multiple phases such as scanning logs, transferring state, and changing roles.

\begin{figure}[b]
  \centering
  \includegraphics[width=\columnwidth]{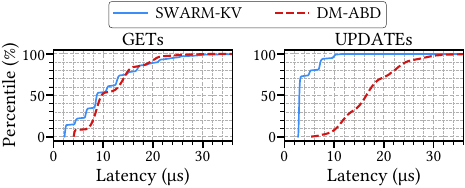}
  \caption{Latency CDFs for \kvs and DM-ABD with YCSB workload A, a single key-value pair, and 16 clients.} \label{fig:eval-contention}
\end{figure}

\subsection{Extreme Contention} \label{sec:eval:contention}

We evaluate the impact of extreme contention on \kvs's performance
  by studying the latency of \getop{}s and \updateop{}s when either 16 or 32 clients
  hit a single key.
Figure~\ref{fig:eval-contention} shows the
  results under YCSB workload A with 16 clients.

With 16 clients, \getop{}s remain live but the latency of their 99th percentile degrades to 30\,\us.
We see that only 14\% of \getop{}s complete in 1 \rt by finding a correct in-place value.
8\% of \getop{}s complete in 2 \rts by reading the value out of place.
The next 78\% of \getop{}s, however, complete in more than 2 \rts as they must either iterate a few times before identifying a valid value, or reading \protocol' max register requires a writing step (\S\ref{sec:background:abd}).
The latency of DM-ABD's \getop{}s also significantly increases because of the latter, but it is drastically exacerbated by the 
contention on the \maxop substitute, which lacks \kvs's contention-reduction mechanism (\S\ref{sec:ino:max}).
 
Meanwhile, \updateop{}s complete in at most 4 \rts, with their 99th percentile reaching 10\,\us.
73\% of \updateop{}s complete in 1 \rt by guessing a fresh timestamp and identifying it as so immediately.
7\% of \updateop{}s complete in 2 \rts by guessing a fresh timestamp and using their timestamp lock to detect it as fresh.
14\% of \updateop{}s complete in 3 \rts by also having to either (a) replicate the value of another writer to a majority or (b) retry their \writeop with a fresh timestamp.
6\% of \updateop{}s complete in 4 \rts by having to do both (a) and (b).
Similarly to \getop{}s, the latency of DM-ABD \updateop{}s significantly degrades due to high contention on the unreliable max registers.

With 32 clients (not shown), \kvs's \getop{}s see their median latency degrade by 4\,\us while \updateop{}s are mostly unaffected.
DM-ABD, however, is severely impacted as the median latency of its \getop{}s and \updateop{}s both reach 47\,\us.
 
\subsection{Scalability of our \maxop Substitute} \label{sec:eval:micro}
We now evaluate how varying the number of
  \innout 8\,B metadata buffers (\S\ref{sec:ino:max}) affects 
  the latency of \kvs by running YCSB workloads A and B with 64 clients.
  Figure~\ref{fig:eval-max} shows the results for workload B.

\begin{figure}
  \centering
  \includegraphics[width=\columnwidth]{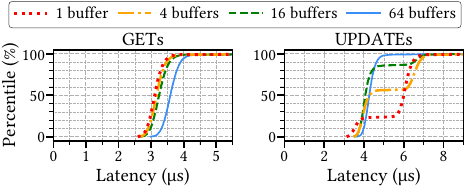}
  \caption{Latency CDFs for \kvs with YCSB workload B, 64 clients, and a varying number of metadata buffers.} \label{fig:eval-max}
\end{figure}

The use of multiple metadata buffers per key significantly reduces the latency of \updateop{}s but slightly increases the latency of \getop{}s.
With a single metadata buffer per key that is shared among all
  writers, only 23\% of \updateop{}s complete in one \rt
  by having accessed the key since its last \updateop{}.
With 4 buffers per key, each shared by 16 clients,
  merely 57\% of \updateop{}s complete in one \rt.
This raises to 86\% with 16 buffers per key.
Lastly, with one metadata buffer per key and per writer, 99\% of \updateop{}s finish in 1 \rt, the last percent resulting from concurrency or from the guess of stale timestamps.

In the meantime, increasing the number of buffers increases the latency of \getop{}s and single-\rt \updateop{}s because they must read larger arrays of buffers.
With 1 buffer, \getop{}s and single-\rt \updateop{}s have a median latency of 3.1\,\us and 3.6\,\us, respectively.
These latencies grow to 3.2\,\us and 3.9\,\us with 4 buffers, 3.3\,\us and 4\,\us with 16 buffers, and 3.6\,\us and 4.3\,\us with 64 buffers.

Under YCSB workload A (not shown) \getop{}s have the same latency, but only 2\% of \updateop{}s complete in one \rt with 1 buffer, 11\% with 4, 39\% with 16, and 99\% with 64.
We thus recommend using 1 buffer per client as it has low impact on the latency of 1-\rt operations (+15\% for 64 clients), but makes 1-\rt \updateop{}s much more frequent.
 However, this means that the more clients there are, the larger the arrays of metadata buffers become, which hinders performance and limits scalability.
\section{Related Work} \label{sec:related}

Our work touches on various topics:
replication, 
concurrency control, key-value stores and disaggregated memory.

\paragraph{Replication.}
Linearizable \readop{}-\writeop{} replication has been the subject of active research. 
The seminal ABD~\cite{abd-emulation} protocol is the first
of a long line of protocols~\cite{dutta2004howfast, burkhard2009efficiency, huang2020finegrained, tseng2023distributed}.
Dutta \textit{et al.}~\cite{dutta2004howfast}
prove that one cannot always have one-\rt
\readop{}s or \writeop{}s,
so several researchers~\cite{georgiou2008robustness, georgiou2009semifast, konwar2020semifast} investigated semi-fast implementations, where \writeop{}s take one \rt, but
\readop{}s take longer.

Similarly to \sysname, Gus~\cite{tseng2023distributed} achieves single-\rt latency
for both operations in the common case. However, Gus fundamentally differs
from \sysname in the extra assumptions it makes and the techniques it uses.
Specifically: (1) Gus assumes that replicas are regular machines capable of running arbitrary RPCs; in contrast, we consider memory replicas that are operated on by one-sided RDMA, which is a challenge
we address.
(2) Gus requires each client to be co-hosted
with one of the replicas for low latency;
in contrast, our clients are hosts remotely connected to memory nodes, which
makes the problem harder. (3) Gus requires more than 2f+1 replicas to tolerate f failures when f>2; in contrast, we require only 2f+1 replicas. Gus, however, has a theoretical advantage over our proposal: even in the worst case, \readop{}s and \writeop{}s complete in at most 2 RTTs (while \sysname/Safe-Guess may take a slow path in rare cases).

Various consensus-based protocols have been proposed to achieve strongly consistent memory replication. 
EPaxos~\cite{moraru2013epaxos}, Giza~\cite{chen2017giza}, and TEMPO~\cite{enes2021tempo} are recent examples of such protocols.
While these protocols can replicate any data object, 
\sysname considers only \readop{}/\writeop{} operations. 
However, consensus-based replication protocols typically have a slower fast path than \sysname (i.e., requiring more roundtrips),
and incur long delays when synchrony assumptions do not hold (e.g., when timeouts are exceeded).
Gryff~\cite{burke2020gryff} is an interesting consensus-based system
which combines consensus and ABD: it provides \readop{}/\writeop{} operations without utilizing consensus,
and resorts to consensus for other operations.
Although Gryff matches \sysname's single-\rt
\readop{}s, it requires two \rt{}s for \writeop{}s.
Importantly, none of the aforementioned systems considers the disaggregated memory setting where replicas sit on nodes without general compute capability. 
Carbink~\cite{zhou2022carbink} implements reliable disaggregated memory, but considering a non-shared access: only one client accesses
a given memory location. 

\paragraph{Concurrency control.}
Some of \sysname's techniques to achieve single-\rt 
\readop{}/\writeop{}
latency 
are inspired from other systems.
Utilizing real-time timestamps for optimistic ordering appears in NCC~\cite{lu2023ncc} to efficiently order transactions.
Self-validating \readop{}s appear in both Pilaf~\cite{mitchell2013pilaf} and FARM~\cite{dragojevic2014farm} to avoid inconsistencies stemming from RDMA's weak \readop{} guarantees.
Out-of-place updates appear in Clover~\cite{tsai2020clover}, which uses them to achieve concurrent lock-free access to disaggregated memory.
Sherman~\cite{sherman2022} leverages the FIFO guarantee of RDMA to minimize memory \rts when
modifying a lock-protected tree structure. 

\paragraph{Disaggregated KVS.}
Several recent key-value stores have been designed for disaggregated memory, e.g., DINOMO~\cite{lee2022dinomo},
  AsymNVM~\cite{ma2020asymnvm}, Clover~\cite{tsai2020clover}, ROLEX~\cite{li2023rolex}, and FUSEE~\cite{fusee}.
These works differ from \kvs in various ways.
DINOMO requires clients to communicate with
  intermediary key-value servers
  instead of directly accessing memory nodes, which requires at least two \rts to update data.
DINOMO, AsymNVM, and ROLEX make different assumptions from us: they
  require memory nodes to have some compute
  to offload custom functionality of the key-value 
  store.
DINOMO, AsymNVM, and Clover have a different goal from us:
  they aim to provide durability, which require persistent
  memory and different techniques.
Moreover, Clover is not designed for scalability since
  it uses a metadata server in the control plane,
  which becomes a bottleneck.
Meanwhile, ROLEX focuses on improving the index performance
  through learning. 
The closest related work to \kvs is FUSEE, 
  which allows clients to access disaggregated memory directly,
  requires no compute at memory nodes,
  provides scalability,
  and targets volatile memory.
  FUSEE heavily relies on caching to keep latency low, but its \getop{}s often finish in 2 \rts, and its \updateop{}s take at least 4 \rt{}s, resulting in lower performance, as shown by our evaluation.
\section{Concluding Remarks} \label{sec:conclusion}

\sysname is the first replication scheme for in-disaggregated-memory shared objects to provide
  (1) single-\rt \readop{}s and \writeop{}s in the common case,
  (2) strong consistency (linearizability), and
  (3) strong liveness (wait-freedom).
  \sysname makes two independent contributions.
The first, \protocol, is the first such protocol assuming memory
  nodes can provide atomic conditional updates to large objects
  in a single \rt.
The second, \innout, is a novel technique to obtain
  the atomic conditional update 
  of large buffers in one \rt without requiring compute at memory nodes.
We showed the utility of \sysname
  by building \kvs, a low-latency, strongly consistent
  and highly available disaggregated key-value store
  that significantly outperforms the state of the art
  and incurs little replication latency overhead.

We foresee several extensions to our work. 
While
\sysname considers in-memory replication, adding durability
without harming performance
is appealing, but appears non-trivial~\cite{ganesan2021orca,leblanc2023chipmunk}.
Also, formally proving \sysname's algorithms, but also its implementation (along the lines of IronFleet~\cite{hawblitzel2015ironfleet}), is interesting future work.

\section*{Acknowledgments} \label{sec:acknowledgments}
We thank the anonymous reviewers and our shepherd Wyatt Lloyd
for their valuable comments,
as well as the anonymous artifact evaluators for reviewing our implementation.
We also thank Thomas Bourgeat for his feedback.

\bibliographystyle{ACM-Reference-Format}
\bibliography{references}


\begin{thebibliography}{52}


\ifx \showCODEN    \undefined \def \showCODEN     #1{\unskip}     \fi
\ifx \showDOI      \undefined \def \showDOI       #1{#1}\fi
\ifx \showISBNx    \undefined \def \showISBNx     #1{\unskip}     \fi
\ifx \showISBNxiii \undefined \def \showISBNxiii  #1{\unskip}     \fi
\ifx \showISSN     \undefined \def \showISSN      #1{\unskip}     \fi
\ifx \showLCCN     \undefined \def \showLCCN      #1{\unskip}     \fi
\ifx \shownote     \undefined \def \shownote      #1{#1}          \fi
\ifx \showarticletitle \undefined \def \showarticletitle #1{#1}   \fi
\ifx \showURL      \undefined \def \showURL       {\relax}        \fi
\providecommand\bibfield[2]{#2}
\providecommand\bibinfo[2]{#2}
\providecommand\natexlab[1]{#1}
\providecommand\showeprint[2][]{arXiv:#2}

\bibitem[Anderson et~al\mbox{.}(2021)]%
        {wfhp}
\bibfield{author}{\bibinfo{person}{Daniel Anderson}, \bibinfo{person}{Guy~E. Blelloch}, {and} \bibinfo{person}{Yuanhao Wei}.} \bibinfo{year}{2021}\natexlab{}.
\newblock \showarticletitle{Concurrent deferred reference counting with constant-time overhead}. In \bibinfo{booktitle}{\emph{Proceedings of the 42nd ACM SIGPLAN International Conference on Programming Language Design and Implementation}} (Virtual Event) \emph{(\bibinfo{series}{PLDI '21})}. \bibinfo{publisher}{Association for Computing Machinery}, \bibinfo{address}{New York, NY, USA}, \bibinfo{pages}{526–541}.
\newblock
\showISBNx{9781450383912}
\urldef\tempurl%
\url{https://doi.org/10.1145/3453483.3454060}
\showURL{%
\tempurl}


\bibitem[Aspnes et~al\mbox{.}(2009)]%
        {maxregister}
\bibfield{author}{\bibinfo{person}{James Aspnes}, \bibinfo{person}{Hagit Attiya}, {and} \bibinfo{person}{Keren Censor}.} \bibinfo{year}{2009}\natexlab{}.
\newblock \showarticletitle{Max registers, counters, and monotone circuits}. In \bibinfo{booktitle}{\emph{Proceedings of the 28th {ACM} Symposium on Principles of Distributed Computing}} (Calgary, Alberta, Canada) \emph{(\bibinfo{series}{PODC '09})}. \bibinfo{publisher}{Association for Computing Machinery}, \bibinfo{address}{New York, NY, USA}, \bibinfo{pages}{36--45}.
\newblock
\urldef\tempurl%
\url{https://doi.org/10.1145/1582716.1582728}
\showDOI{\tempurl}


\bibitem[Aspnes and Ellen(2014)]%
        {aspens2014tight}
\bibfield{author}{\bibinfo{person}{James Aspnes} {and} \bibinfo{person}{Faith Ellen}.} \bibinfo{year}{2014}\natexlab{}.
\newblock \showarticletitle{Tight Bounds for Adopt-Commit Objects}.
\newblock \bibinfo{journal}{\emph{Theory of Computing Systems}} \bibinfo{volume}{55}, \bibinfo{number}{3} (\bibinfo{date}{Oct.} \bibinfo{year}{2014}), \bibinfo{pages}{451–474}.
\newblock
\showISSN{1432-4350}
\urldef\tempurl%
\url{https://doi.org/10.1007/s00224-013-9448-1}
\showDOI{\tempurl}


\bibitem[Attiya et~al\mbox{.}(1995)]%
        {abd-emulation}
\bibfield{author}{\bibinfo{person}{Hagit Attiya}, \bibinfo{person}{Amotz Bar-Noy}, {and} \bibinfo{person}{Danny Dolev}.} \bibinfo{year}{1995}\natexlab{}.
\newblock \showarticletitle{Sharing Memory Robustly in Message-Passing Systems}.
\newblock \bibinfo{journal}{\emph{J. ACM}} \bibinfo{volume}{42}, \bibinfo{number}{1} (\bibinfo{date}{Jan.} \bibinfo{year}{1995}), \bibinfo{pages}{124–142}.
\newblock
\showISSN{0004-5411}
\urldef\tempurl%
\url{https://doi.org/10.1145/200836.200869}
\showDOI{\tempurl}


\bibitem[Attiya and Welch(2004)]%
        {attiyawelch}
\bibfield{author}{\bibinfo{person}{Hagit Attiya} {and} \bibinfo{person}{Jennifer Welch}.} \bibinfo{year}{2004}\natexlab{}.
\newblock \bibinfo{booktitle}{\emph{Distributed Computing: Fundamentals, Simulations and Advanced Topics}}.
\newblock \bibinfo{publisher}{John Wiley and Sons, Inc.}, \bibinfo{address}{Hoboken, NJ, USA}.
\newblock
\showISBNx{0471453242}


\bibitem[Beck and Kagan(2011)]%
        {beck2011roce-performance}
\bibfield{author}{\bibinfo{person}{Motti Beck} {and} \bibinfo{person}{Michael Kagan}.} \bibinfo{year}{2011}\natexlab{}.
\newblock \showarticletitle{Performance Evaluation of the {RDMA} over Ethernet ({RoCE}) Standard in Enterprise Data Centers Infrastructure}. In \bibinfo{booktitle}{\emph{Proceedings of the 3rd Workshop on Data Center - Converged and Virtual Ethernet Switching}} (San Francisco, CA, USA) \emph{(\bibinfo{series}{DC-CaVES '11})}. \bibinfo{publisher}{International Teletraffic Congress}, \bibinfo{address}{San Francisco, CA, USA}, \bibinfo{pages}{9–15}.
\newblock
\showISBNx{9780983628323}
\urldef\tempurl%
\url{https://dl.acm.org/doi/10.5555/2043535.2043537}
\showURL{%
\tempurl}


\bibitem[Burke et~al\mbox{.}(2020)]%
        {burke2020gryff}
\bibfield{author}{\bibinfo{person}{Matthew Burke}, \bibinfo{person}{Audrey Cheng}, {and} \bibinfo{person}{Wyatt Lloyd}.} \bibinfo{year}{2020}\natexlab{}.
\newblock \showarticletitle{Gryff: Unifying Consensus and Shared Registers}. In \bibinfo{booktitle}{\emph{Proceedings of the 17th USENIX Symposium on Networked Systems Design and Implementation}} (Santa Clara, CA, USA) \emph{(\bibinfo{series}{NSDI '20})}. \bibinfo{publisher}{USENIX Association}, \bibinfo{address}{Berkeley, CA, USA}, \bibinfo{pages}{591--617}.
\newblock
\showISBNx{978-1-939133-13-7}
\urldef\tempurl%
\url{https://www.usenix.org/conference/nsdi20/presentation/burke}
\showURL{%
\tempurl}


\bibitem[Chen et~al\mbox{.}(2017)]%
        {chen2017giza}
\bibfield{author}{\bibinfo{person}{Yu~Lin Chen}, \bibinfo{person}{Shuai Mu}, \bibinfo{person}{Jinyang Li}, \bibinfo{person}{Cheng Huang}, \bibinfo{person}{Jin Li}, \bibinfo{person}{Aaron Ogus}, {and} \bibinfo{person}{Douglas Phillips}.} \bibinfo{year}{2017}\natexlab{}.
\newblock \showarticletitle{Giza: Erasure Coding Objects across Global Data Centers}. In \bibinfo{booktitle}{\emph{Proceedings of the 2017 USENIX Annual Technical Conference}} (Santa Clara, CA, USA) \emph{(\bibinfo{series}{USENIX ATC '17})}. \bibinfo{publisher}{USENIX Association}, \bibinfo{address}{Berkeley, CA, USA}, \bibinfo{pages}{539--551}.
\newblock
\showISBNx{978-1-931971-38-6}
\urldef\tempurl%
\url{https://www.usenix.org/conference/atc17/technical-sessions/presentation/chen-yu-lin}
\showURL{%
\tempurl}


\bibitem[Collet(2022)]%
        {xxhash}
\bibfield{author}{\bibinfo{person}{Yann Collet}.} \bibinfo{year}{2022}\natexlab{}.
\newblock \bibinfo{title}{{xxHash}: Extremely fast non-cryptographic hash algorithm}.
\newblock
\newblock
\urldef\tempurl%
\url{https://github.com/Cyan4973/xxHash}
\showURL{%
\tempurl}
\newblock
\shownote{Accessed 2024-03-17}.


\bibitem[Consortium(2022)]%
        {intelcxl}
\bibfield{author}{\bibinfo{person}{Compute Express~Link Consortium}.} \bibinfo{year}{2022}\natexlab{}.
\newblock \bibinfo{title}{{Compute} {Express} {Link} ({CXL}) Specification, Revision 3.0}.
\newblock
\newblock
\urldef\tempurl%
\url{https://www.computeexpresslink.org/}
\showURL{%
\tempurl}
\newblock
\shownote{Accessed 2024-04-13}.


\bibitem[Cooper et~al\mbox{.}(2010)]%
        {ycsb}
\bibfield{author}{\bibinfo{person}{Brian~F. Cooper}, \bibinfo{person}{Adam Silberstein}, \bibinfo{person}{Erwin Tam}, \bibinfo{person}{Raghu Ramakrishnan}, {and} \bibinfo{person}{Russell Sears}.} \bibinfo{year}{2010}\natexlab{}.
\newblock \showarticletitle{Benchmarking cloud serving systems with {YCSB}}. In \bibinfo{booktitle}{\emph{Proceedings of the 1st ACM Symposium on Cloud Computing}} (Indianapolis, IN, USA) \emph{(\bibinfo{series}{SoCC '10})}. \bibinfo{publisher}{Association for Computing Machinery}, \bibinfo{address}{New York, NY, USA}, \bibinfo{pages}{143–154}.
\newblock
\showISBNx{9781450300360}
\urldef\tempurl%
\url{https://doi.org/10.1145/1807128.1807152}
\showDOI{\tempurl}


\bibitem[Danial(2022)]%
        {cloc}
\bibfield{author}{\bibinfo{person}{Al Danial}.} \bibinfo{year}{2022}\natexlab{}.
\newblock \bibinfo{title}{{CLOC: Count Lines of Code}}.
\newblock
\newblock
\urldef\tempurl%
\url{https://github.com/AlDanial/cloc}
\showURL{%
\tempurl}


\bibitem[Dragojevi\'{c} et~al\mbox{.}(2014)]%
        {dragojevic2014farm}
\bibfield{author}{\bibinfo{person}{Aleksandar Dragojevi\'{c}}, \bibinfo{person}{Dushyanth Narayanan}, \bibinfo{person}{Orion Hodson}, {and} \bibinfo{person}{Miguel Castro}.} \bibinfo{year}{2014}\natexlab{}.
\newblock \showarticletitle{{FaRM}: Fast Remote Memory}. In \bibinfo{booktitle}{\emph{Proceedings of the 11th USENIX Symposium on Networked Systems Design and Implementation}} (Seattle, WA, USA) \emph{(\bibinfo{series}{NSDI '14})}. \bibinfo{publisher}{USENIX Association}, \bibinfo{address}{Berkeley, CA, USA}, \bibinfo{pages}{401–414}.
\newblock
\showISBNx{9781931971096}
\urldef\tempurl%
\url{https://www.usenix.org/conference/nsdi14/technical-sessions/dragojevi\%C4\%87}
\showURL{%
\tempurl}


\bibitem[Dutta et~al\mbox{.}(2004)]%
        {dutta2004howfast}
\bibfield{author}{\bibinfo{person}{Partha Dutta}, \bibinfo{person}{Rachid Guerraoui}, \bibinfo{person}{Ron~R. Levy}, {and} \bibinfo{person}{Arindam Chakraborty}.} \bibinfo{year}{2004}\natexlab{}.
\newblock \showarticletitle{How fast can a distributed atomic read be?}. In \bibinfo{booktitle}{\emph{Proceedings of the 23rd {ACM} Symposium on Principles of Distributed Computing}} (St. John's, Newfoundland, Canada) \emph{(\bibinfo{series}{PODC '04})}. \bibinfo{publisher}{Association for Computing Machinery}, \bibinfo{address}{New York, NY, USA}, \bibinfo{pages}{236–245}.
\newblock
\showISBNx{1581138024}
\urldef\tempurl%
\url{https://doi.org/10.1145/1011767.1011802}
\showDOI{\tempurl}


\bibitem[Enes et~al\mbox{.}(2021)]%
        {enes2021tempo}
\bibfield{author}{\bibinfo{person}{Vitor Enes}, \bibinfo{person}{Carlos Baquero}, \bibinfo{person}{Alexey Gotsman}, {and} \bibinfo{person}{Pierre Sutra}.} \bibinfo{year}{2021}\natexlab{}.
\newblock \showarticletitle{Efficient replication via timestamp stability}. In \bibinfo{booktitle}{\emph{Proceedings of the 16th European Conference on Computer Systems}} (Virtual Event) \emph{(\bibinfo{series}{EuroSys '21})}. \bibinfo{publisher}{Association for Computing Machinery}, \bibinfo{address}{New York, NY, USA}, \bibinfo{pages}{178–193}.
\newblock
\showISBNx{9781450383349}
\urldef\tempurl%
\url{https://doi.org/10.1145/3447786.3456236}
\showDOI{\tempurl}


\bibitem[Englert et~al\mbox{.}(2009)]%
        {burkhard2009efficiency}
\bibfield{author}{\bibinfo{person}{Burkhard Englert}, \bibinfo{person}{Chryssis Georgiou}, \bibinfo{person}{Peter~M. Musial}, \bibinfo{person}{Nicolas~C. Nicolaou}, {and} \bibinfo{person}{Alexander~A. Shvartsman}.} \bibinfo{year}{2009}\natexlab{}.
\newblock \showarticletitle{On the Efficiency of Atomic Multi-reader, Multi-writer Distributed Memory}. In \bibinfo{booktitle}{\emph{13th International Conference on Principles of Distributed Systems}} (N{\^{\i}}mes, France) \emph{(\bibinfo{series}{OPODIS '09'})}. \bibinfo{publisher}{Springer-Verlag}, \bibinfo{address}{Berlin, Germany}, \bibinfo{pages}{240--254}.
\newblock
\urldef\tempurl%
\url{https://doi.org/10.1007/978-3-642-10877-8\_20}
\showDOI{\tempurl}


\bibitem[Fischer et~al\mbox{.}(1985)]%
        {fischer1985flp}
\bibfield{author}{\bibinfo{person}{Michael~J. Fischer}, \bibinfo{person}{Nancy~A. Lynch}, {and} \bibinfo{person}{Michael~S. Paterson}.} \bibinfo{year}{1985}\natexlab{}.
\newblock \showarticletitle{Impossibility of Distributed Consensus with One Faulty Process}.
\newblock \bibinfo{journal}{\emph{J. ACM}} \bibinfo{volume}{32}, \bibinfo{number}{2} (\bibinfo{date}{apr} \bibinfo{year}{1985}), \bibinfo{pages}{374–382}.
\newblock
\showISSN{0004-5411}
\urldef\tempurl%
\url{https://doi.org/10.1145/3149.214121}
\showDOI{\tempurl}


\bibitem[Ganesan et~al\mbox{.}(2021)]%
        {ganesan2021orca}
\bibfield{author}{\bibinfo{person}{Aishwarya Ganesan}, \bibinfo{person}{Ramnatthan Alagappan}, \bibinfo{person}{Andrea~C. Arpaci-Dusseau}, {and} \bibinfo{person}{Remzi~H. Arpaci-Dusseau}.} \bibinfo{year}{2021}\natexlab{}.
\newblock \showarticletitle{Strong and Efficient Consistency with Consistency-aware Durability}.
\newblock \bibinfo{journal}{\emph{{ACM} Transactions on Storage}} \bibinfo{volume}{17}, \bibinfo{number}{1}, Article \bibinfo{articleno}{4} (\bibinfo{date}{Jan.} \bibinfo{year}{2021}), \bibinfo{numpages}{27}~pages.
\newblock
\showISSN{1553-3077}
\urldef\tempurl%
\url{https://doi.org/10.1145/3423138}
\showDOI{\tempurl}


\bibitem[Georgiou et~al\mbox{.}(2008)]%
        {georgiou2008robustness}
\bibfield{author}{\bibinfo{person}{Chryssis Georgiou}, \bibinfo{person}{Nicolas~C. Nicolaou}, {and} \bibinfo{person}{Alexander~A. Shvartsman}.} \bibinfo{year}{2008}\natexlab{}.
\newblock \showarticletitle{On the Robustness of (Semi) Fast Quorum-Based Implementations of Atomic Shared Memory}. In \bibinfo{booktitle}{\emph{Proceedings of the 22nd International Symposium on Distributed Computing}} (Arcachon, France) \emph{(\bibinfo{series}{DISC '08}, Vol.~\bibinfo{volume}{5218})}. \bibinfo{publisher}{Springer-Verlag}, \bibinfo{address}{Berlin, Germany}, \bibinfo{pages}{289--304}.
\newblock
\urldef\tempurl%
\url{https://doi.org/10.1007/978-3-540-87779-0\_20}
\showDOI{\tempurl}


\bibitem[Georgiou et~al\mbox{.}(2009)]%
        {georgiou2009semifast}
\bibfield{author}{\bibinfo{person}{Chryssis Georgiou}, \bibinfo{person}{Nicolas~C. Nicolaou}, {and} \bibinfo{person}{Alexander~A. Shvartsman}.} \bibinfo{year}{2009}\natexlab{}.
\newblock \showarticletitle{Fault-tolerant semifast implementations of atomic read/write registers}.
\newblock \bibinfo{journal}{\emph{J. Parallel and Distrib. Comput.}} \bibinfo{volume}{69}, \bibinfo{number}{1} (\bibinfo{year}{2009}), \bibinfo{pages}{62--79}.
\newblock
\urldef\tempurl%
\url{https://doi.org/10.1016/J.JPDC.2008.05.004}
\showDOI{\tempurl}


\bibitem[Gu et~al\mbox{.}(2017)]%
        {gu2017infiniswap}
\bibfield{author}{\bibinfo{person}{Juncheng Gu}, \bibinfo{person}{Youngmoon Lee}, \bibinfo{person}{Yiwen Zhang}, \bibinfo{person}{Mosharaf Chowdhury}, {and} \bibinfo{person}{Kang~G. Shin}.} \bibinfo{year}{2017}\natexlab{}.
\newblock \showarticletitle{Efficient Memory Disaggregation with Infiniswap}. In \bibinfo{booktitle}{\emph{Proceedings of the 14th USENIX Symposium on Networked Systems Design and Implementation}} (Boston, MA, USA). \bibinfo{publisher}{{USENIX} Association}, \bibinfo{address}{Berkeley, CA, USA}, \bibinfo{pages}{649--667}.
\newblock
\urldef\tempurl%
\url{https://www.usenix.org/conference/nsdi17/technical-sessions/presentation/gu}
\showURL{%
\tempurl}


\bibitem[Guerraoui et~al\mbox{.}(2022)]%
        {ukharon}
\bibfield{author}{\bibinfo{person}{Rachid Guerraoui}, \bibinfo{person}{Antoine Murat}, \bibinfo{person}{Javier Picorel}, \bibinfo{person}{Athanasios Xygkis}, \bibinfo{person}{Huabing Yan}, {and} \bibinfo{person}{Pengfei Zuo}.} \bibinfo{year}{2022}\natexlab{}.
\newblock \showarticletitle{{uKharon}: {A} Membership Service for Microsecond Applications}. In \bibinfo{booktitle}{\emph{Proceedings of the 2022 USENIX Annual Technical Conference}} (Carlsbad, CA, USA) \emph{(\bibinfo{series}{USENIX ATC '22})}. \bibinfo{publisher}{USENIX Association}, \bibinfo{address}{Berkeley, CA, USA}, \bibinfo{pages}{101--120}.
\newblock
\urldef\tempurl%
\url{https://www.usenix.org/conference/atc22/presentation/guerraoui}
\showURL{%
\tempurl}


\bibitem[Hat(2020)]%
        {tsc}
\bibfield{author}{\bibinfo{person}{Red Hat}.} \bibinfo{year}{2020}\natexlab{}.
\newblock \bibinfo{title}{{RHEL} for Real Time Timestamping}.
\newblock
\newblock
\urldef\tempurl%
\url{https://access.redhat.com/documentation/en-us/red_hat_enterprise_linux_for_real_time/7/html/reference_guide/chap-timestamping}
\showURL{%
\tempurl}
\newblock
\shownote{Accessed: April 14, 2023}.


\bibitem[Hawblitzel et~al\mbox{.}(2015)]%
        {hawblitzel2015ironfleet}
\bibfield{author}{\bibinfo{person}{Chris Hawblitzel}, \bibinfo{person}{Jon Howell}, \bibinfo{person}{Manos Kapritsos}, \bibinfo{person}{Jacob~R. Lorch}, \bibinfo{person}{Bryan Parno}, \bibinfo{person}{Michael~L. Roberts}, \bibinfo{person}{Srinath Setty}, {and} \bibinfo{person}{Brian Zill}.} \bibinfo{year}{2015}\natexlab{}.
\newblock \showarticletitle{{IronFleet}: proving practical distributed systems correct}. In \bibinfo{booktitle}{\emph{Proceedings of the 25th ACM Symposium on Operating Systems Principles}} (Monterey, CA, USA) \emph{(\bibinfo{series}{SOSP '15})}. \bibinfo{publisher}{Association for Computing Machinery}, \bibinfo{address}{New York, NY, USA}, \bibinfo{pages}{1–17}.
\newblock
\showISBNx{9781450338349}
\urldef\tempurl%
\url{https://doi.org/10.1145/2815400.2815428}
\showDOI{\tempurl}


\bibitem[Herlihy(1991)]%
        {herlihy1991-waitfree}
\bibfield{author}{\bibinfo{person}{Maurice Herlihy}.} \bibinfo{year}{1991}\natexlab{}.
\newblock \showarticletitle{Wait-Free Synchronization}.
\newblock \bibinfo{journal}{\emph{ACM Transactions on Programming Languages and Systems}} \bibinfo{volume}{13}, \bibinfo{number}{1} (\bibinfo{date}{Jan.} \bibinfo{year}{1991}), \bibinfo{pages}{124–149}.
\newblock
\showISSN{0164-0925}
\urldef\tempurl%
\url{https://doi.org/10.1145/114005.102808}
\showDOI{\tempurl}


\bibitem[Herlihy and Wing(1990)]%
        {herlihy1990-linearizability}
\bibfield{author}{\bibinfo{person}{Maurice Herlihy} {and} \bibinfo{person}{Jeannette Wing}.} \bibinfo{year}{1990}\natexlab{}.
\newblock \showarticletitle{Linearizability: a correctness condition for concurrent objects}.
\newblock \bibinfo{journal}{\emph{ACM Transactions on Programming Languages and Systems}} \bibinfo{volume}{12}, \bibinfo{number}{3} (\bibinfo{date}{July} \bibinfo{year}{1990}), \bibinfo{pages}{463–492}.
\newblock
\showISSN{0164-0925}
\urldef\tempurl%
\url{https://doi.org/10.1145/78969.78972}
\showDOI{\tempurl}


\bibitem[Huang et~al\mbox{.}(2020)]%
        {huang2020finegrained}
\bibfield{author}{\bibinfo{person}{Kaile Huang}, \bibinfo{person}{Yu Huang}, {and} \bibinfo{person}{Hengfeng Wei}.} \bibinfo{year}{2020}\natexlab{}.
\newblock \showarticletitle{Fine-grained Analysis on Fast Implementations of Distributed Multi-writer Atomic Registers}. In \bibinfo{booktitle}{\emph{Proceedings of the 39th {ACM} Symposium on Principles of Distributed Computing}} (Virtual Event) \emph{(\bibinfo{series}{PODC '20})}. \bibinfo{publisher}{Association for Computing Machinery}, \bibinfo{address}{New York, NY, USA}, \bibinfo{pages}{200–209}.
\newblock
\showISBNx{9781450375825}
\urldef\tempurl%
\url{https://doi.org/10.1145/3382734.3405698}
\showDOI{\tempurl}


\bibitem[Hunt et~al\mbox{.}(2010)]%
        {hunt2010zookeeper}
\bibfield{author}{\bibinfo{person}{Patrick Hunt}, \bibinfo{person}{Mahadev Konar}, \bibinfo{person}{Flavio~P. Junqueira}, {and} \bibinfo{person}{Benjamin Reed}.} \bibinfo{year}{2010}\natexlab{}.
\newblock \showarticletitle{ZooKeeper: Wait-Free Coordination for Internet-Scale Systems}. In \bibinfo{booktitle}{\emph{Proceedings of the 2010 USENIX Annual Technical Conference}} (Boston, MA, USA) \emph{(\bibinfo{series}{USENIX ATC '10})}. \bibinfo{publisher}{USENIX Association}, \bibinfo{address}{Berkeley, CA, USA}, \bibinfo{numpages}{11}~pages.
\newblock
\urldef\tempurl%
\url{https://www.usenix.org/conference/usenix-atc-10/zookeeper-wait-free-coordination-internet-scale-systems}
\showURL{%
\tempurl}


\bibitem[Jain et~al\mbox{.}(2024)]%
        {cxl-atomicity}
\bibfield{author}{\bibinfo{person}{Sunita Jain}, \bibinfo{person}{Nagaradhesh Yeleswarapu}, \bibinfo{person}{Hasan~Al Maruf}, {and} \bibinfo{person}{Rita Gupta}.} \bibinfo{year}{2024}\natexlab{}.
\newblock \showarticletitle{Memory Sharing with {CXL}: Hardware and Software Design Approaches}. In \bibinfo{booktitle}{\emph{Proceedings of the 3rd Workshop on Heterogeneous Composable and Disaggregated Systems}} (San Diego, CA, USA) \emph{(\bibinfo{series}{HCDS '24})}. \bibinfo{publisher}{Association for Computing Machinery}, \bibinfo{address}{New York, NY, USA}, \bibinfo{numpages}{4}~pages.
\newblock
\urldef\tempurl%
\url{https://arxiv.org/pdf/2404.03245.pdf}
\showURL{%
\tempurl}


\bibitem[Jin et~al\mbox{.}(2002)]%
        {pfister2001infiniband}
\bibfield{author}{\bibinfo{person}{Hai Jin}, \bibinfo{person}{Rajkumar Buyya}, {and} \bibinfo{person}{Toni Cortes}.} \bibinfo{year}{2002}\natexlab{}.
\newblock \showarticletitle{An Introduction to the InfiniBand Architecture}. In \bibinfo{booktitle}{\emph{High Performance Mass Storage and Parallel I/O: Technologies and Applications} (\bibinfo{edition}{1st} ed.)}. \bibinfo{publisher}{John Wiley and Sons, Inc.}, \bibinfo{address}{Hoboken, NJ, USA}, \bibinfo{pages}{616--632}.
\newblock
\urldef\tempurl%
\url{https://doi.org/10.1109/9780470544839}
\showDOI{\tempurl}


\bibitem[Kalia et~al\mbox{.}(2019)]%
        {kalia2019erpc}
\bibfield{author}{\bibinfo{person}{Anuj Kalia}, \bibinfo{person}{Michael Kaminsky}, {and} \bibinfo{person}{David Andersen}.} \bibinfo{year}{2019}\natexlab{}.
\newblock \showarticletitle{Datacenter {RPCs} can be General and Fast}. In \bibinfo{booktitle}{\emph{Proceedings of the 16th USENIX Symposium on Networked Systems Design and Implementation}} (Boston, MA, USA) \emph{(\bibinfo{series}{NSDI '19})}. \bibinfo{publisher}{USENIX Association}, \bibinfo{address}{Berkeley, CA, USA}, \bibinfo{pages}{1--16}.
\newblock
\showISBNx{978-1-931971-49-2}
\urldef\tempurl%
\url{https://www.usenix.org/conference/nsdi19/presentation/kalia}
\showURL{%
\tempurl}


\bibitem[Kalia et~al\mbox{.}(2014)]%
        {kalia2014herd}
\bibfield{author}{\bibinfo{person}{Anuj Kalia}, \bibinfo{person}{Michael Kaminsky}, {and} \bibinfo{person}{David~G. Andersen}.} \bibinfo{year}{2014}\natexlab{}.
\newblock \showarticletitle{Using RDMA Efficiently for Key-Value Services}. In \bibinfo{booktitle}{\emph{Proceedings of the 2014 ACM Conference on SIGCOMM}} (Chicago, IL, USA) \emph{(\bibinfo{series}{SIGCOMM '14})}. \bibinfo{publisher}{Association for Computing Machinery}, \bibinfo{address}{New York, NY, USA}, \bibinfo{pages}{295–306}.
\newblock
\showISBNx{9781450328364}
\urldef\tempurl%
\url{https://doi.org/10.1145/2619239.2626299}
\showDOI{\tempurl}


\bibitem[Konwar et~al\mbox{.}(2020)]%
        {konwar2020semifast}
\bibfield{author}{\bibinfo{person}{Kishori~M. Konwar}, \bibinfo{person}{Saptaparni Kumar}, {and} \bibinfo{person}{Lewis Tseng}.} \bibinfo{year}{2020}\natexlab{}.
\newblock \showarticletitle{Semi-Fast Byzantine-tolerant Shared Register without Reliable Broadcast}. In \bibinfo{booktitle}{\emph{Proceedings of the 2020 {IEEE} 40th International Conference on Distributed Computing Systems}} (Singapore) \emph{(\bibinfo{series}{ICDCS '20})}. \bibinfo{publisher}{{IEEE} Computer Society}, \bibinfo{address}{Los Alamitos, CA, USA}, \bibinfo{pages}{743--753}.
\newblock
\urldef\tempurl%
\url{https://doi.org/10.1109/ICDCS47774.2020.00057}
\showDOI{\tempurl}


\bibitem[Lamport(1978)]%
        {lamport-clocks}
\bibfield{author}{\bibinfo{person}{Leslie Lamport}.} \bibinfo{year}{1978}\natexlab{}.
\newblock \showarticletitle{Time, Clocks, and the Ordering of Events in a Distributed System}.
\newblock \bibinfo{journal}{\emph{Commun. ACM}} \bibinfo{volume}{21}, \bibinfo{number}{7} (\bibinfo{date}{July} \bibinfo{year}{1978}), \bibinfo{pages}{558–565}.
\newblock
\showISSN{0001-0782}
\urldef\tempurl%
\url{https://doi.org/10.1145/359545.359563}
\showDOI{\tempurl}


\bibitem[Lamport(1984)]%
        {lamport-smr}
\bibfield{author}{\bibinfo{person}{Leslie Lamport}.} \bibinfo{year}{1984}\natexlab{}.
\newblock \showarticletitle{Using Time Instead of Timeout for Fault-Tolerant Distributed Systems}.
\newblock \bibinfo{journal}{\emph{ACM Transactions on Programming Languages and Systems}} \bibinfo{volume}{6}, \bibinfo{number}{2} (\bibinfo{date}{April} \bibinfo{year}{1984}), \bibinfo{pages}{254–280}.
\newblock
\showISSN{0164-0925}
\urldef\tempurl%
\url{https://doi.org/10.1145/2993.2994}
\showDOI{\tempurl}


\bibitem[Lamport(1998)]%
        {lamport-paxos}
\bibfield{author}{\bibinfo{person}{Leslie Lamport}.} \bibinfo{year}{1998}\natexlab{}.
\newblock \showarticletitle{The Part-Time Parliament}.
\newblock \bibinfo{journal}{\emph{ACM Transactions on Computer Systems}} \bibinfo{volume}{16}, \bibinfo{number}{2} (\bibinfo{date}{May} \bibinfo{year}{1998}), \bibinfo{pages}{133–169}.
\newblock
\showISSN{0734-2071}
\urldef\tempurl%
\url{https://doi.org/10.1145/279227.279229}
\showDOI{\tempurl}


\bibitem[LeBlanc et~al\mbox{.}(2023)]%
        {leblanc2023chipmunk}
\bibfield{author}{\bibinfo{person}{Hayley LeBlanc}, \bibinfo{person}{Shankara Pailoor}, \bibinfo{person}{Om~Saran K~R~E}, \bibinfo{person}{Isil Dillig}, \bibinfo{person}{James Bornholt}, {and} \bibinfo{person}{Vijay Chidambaram}.} \bibinfo{year}{2023}\natexlab{}.
\newblock \showarticletitle{Chipmunk: Investigating Crash-Consistency in Persistent-Memory File Systems}. In \bibinfo{booktitle}{\emph{Proceedings of the 18th European Conference on Computer Systems}} (Rome, Italy) \emph{(\bibinfo{series}{EuroSys '23})}. \bibinfo{publisher}{Association for Computing Machinery}, \bibinfo{address}{New York, NY, USA}, \bibinfo{pages}{718–733}.
\newblock
\showISBNx{9781450394871}
\urldef\tempurl%
\url{https://doi.org/10.1145/3552326.3567498}
\showDOI{\tempurl}


\bibitem[Lee et~al\mbox{.}(2022)]%
        {lee2022dinomo}
\bibfield{author}{\bibinfo{person}{Sekwon Lee}, \bibinfo{person}{Soujanya Ponnapalli}, \bibinfo{person}{Sharad Singhal}, \bibinfo{person}{Marcos~K. Aguilera}, \bibinfo{person}{Kimberly Keeton}, {and} \bibinfo{person}{Vijay Chidambaram}.} \bibinfo{year}{2022}\natexlab{}.
\newblock \showarticletitle{{DINOMO}: An Elastic, Scalable, High-Performance Key-Value Store for Disaggregated Persistent Memory}.
\newblock \bibinfo{journal}{\emph{Proceedings of the VLDB Endowment}} \bibinfo{volume}{15}, \bibinfo{number}{13} (\bibinfo{date}{Sept.} \bibinfo{year}{2022}), \bibinfo{pages}{4023–4037}.
\newblock
\showISSN{2150-8097}
\urldef\tempurl%
\url{https://doi.org/10.14778/3565838.3565854}
\showDOI{\tempurl}


\bibitem[Li et~al\mbox{.}(2023a)]%
        {pond}
\bibfield{author}{\bibinfo{person}{Huaicheng Li}, \bibinfo{person}{Daniel~S. Berger}, \bibinfo{person}{Lisa Hsu}, \bibinfo{person}{Daniel Ernst}, \bibinfo{person}{Pantea Zardoshti}, \bibinfo{person}{Stanko Novakovic}, \bibinfo{person}{Monish Shah}, \bibinfo{person}{Samir Rajadnya}, \bibinfo{person}{Scott Lee}, \bibinfo{person}{Ishwar Agarwal}, \bibinfo{person}{Mark~D. Hill}, \bibinfo{person}{Marcus Fontoura}, {and} \bibinfo{person}{Ricardo Bianchini}.} \bibinfo{year}{2023}\natexlab{a}.
\newblock \showarticletitle{{Pond}: {CXL}-Based Memory Pooling Systems for Cloud Platforms}. In \bibinfo{booktitle}{\emph{Proceedings of the 28th ACM International Conference on Architectural Support for Programming Languages and Operating Systems, Volume 2}} (Vancouver, BC, Canada) \emph{(\bibinfo{series}{ASPLOS '23'})}. \bibinfo{publisher}{Association for Computing Machinery}, \bibinfo{address}{New York, NY, USA}, \bibinfo{pages}{574–587}.
\newblock
\showISBNx{9781450399166}
\urldef\tempurl%
\url{https://doi.org/10.1145/3575693.3578835}
\showDOI{\tempurl}


\bibitem[Li et~al\mbox{.}(2023b)]%
        {li2023rolex}
\bibfield{author}{\bibinfo{person}{Pengfei Li}, \bibinfo{person}{Yu Hua}, \bibinfo{person}{Pengfei Zuo}, \bibinfo{person}{Zhangyu Chen}, {and} \bibinfo{person}{Jiajie Sheng}.} \bibinfo{year}{2023}\natexlab{b}.
\newblock \showarticletitle{{ROLEX}: A Scalable {RDMA-oriented} Learned {Key-Value} Store for Disaggregated Memory Systems\balance}. In \bibinfo{booktitle}{\emph{Proceedings of the 21st USENIX Conference on File and Storage Technologies}} (Santa Clara, CA, USA) \emph{(\bibinfo{series}{FAST '23})}. \bibinfo{publisher}{USENIX Association}, \bibinfo{address}{Berkeley, CA, USA}, \bibinfo{pages}{99--114}.
\newblock
\showISBNx{978-1-939133-32-8}
\urldef\tempurl%
\url{https://www.usenix.org/conference/fast23/presentation/li-pengfei}
\showURL{%
\tempurl}


\bibitem[Lu et~al\mbox{.}(2023)]%
        {lu2023ncc}
\bibfield{author}{\bibinfo{person}{Haonan Lu}, \bibinfo{person}{Shuai Mu}, \bibinfo{person}{Siddhartha Sen}, {and} \bibinfo{person}{Wyatt Lloyd}.} \bibinfo{year}{2023}\natexlab{}.
\newblock \showarticletitle{{NCC}: Natural Concurrency Control for Strictly Serializable Datastores by Avoiding the Timestamp-Inversion Pitfall}. In \bibinfo{booktitle}{\emph{Proceedings of the 17th USENIX Symposium on Operating Systems Design and Implementation}} (Boston, MA, USA) \emph{(\bibinfo{series}{OSDI '23})}. \bibinfo{publisher}{USENIX Association}, \bibinfo{address}{Berkeley, CA, USA}, \bibinfo{pages}{305--323}.
\newblock
\showISBNx{978-1-939133-34-2}
\urldef\tempurl%
\url{https://www.usenix.org/conference/osdi23/presentation/lu}
\showURL{%
\tempurl}


\bibitem[Lynch and Shvartsman(1997)]%
        {abd-mwmr}
\bibfield{author}{\bibinfo{person}{Nancy~A. Lynch} {and} \bibinfo{person}{Alexander~A. Shvartsman}.} \bibinfo{year}{1997}\natexlab{}.
\newblock \showarticletitle{Robust Emulation of Shared Memory Using Dynamic Quorum-Acknowledged Broadcasts}. In \bibinfo{booktitle}{\emph{Proceedings of the 27th International Symposium on Fault-Tolerant Computing}} (Seattle, WA, USA) \emph{(\bibinfo{series}{FTCS '97})}. \bibinfo{publisher}{{IEEE} Computer Society}, \bibinfo{address}{Los Alamitos, CA, USA}, \bibinfo{pages}{272--281}.
\newblock
\showISBNx{0818678313}
\urldef\tempurl%
\url{https://doi.org/10.1109/FTCS.1997.614100}
\showDOI{\tempurl}


\bibitem[Ma et~al\mbox{.}(2020)]%
        {ma2020asymnvm}
\bibfield{author}{\bibinfo{person}{Teng Ma}, \bibinfo{person}{Mingxing Zhang}, \bibinfo{person}{Kang Chen}, \bibinfo{person}{Zhuo Song}, \bibinfo{person}{Yongwei Wu}, {and} \bibinfo{person}{Xuehai Qian}.} \bibinfo{year}{2020}\natexlab{}.
\newblock \showarticletitle{{AsymNVM}: An Efficient Framework for Implementing Persistent Data Structures on Asymmetric {NVM} Architecture}. In \bibinfo{booktitle}{\emph{Proceedings of the 25gh International Conference on Architectural Support for Programming Languages and Operating Systems}} (Lausanne, Switzerland) \emph{(\bibinfo{series}{ASPLOS '20})}. \bibinfo{publisher}{Association for Computing Machinery}, \bibinfo{address}{New York, NY, USA}, \bibinfo{pages}{757–773}.
\newblock
\showISBNx{9781450371025}
\urldef\tempurl%
\url{https://doi.org/10.1145/3373376.3378511}
\showDOI{\tempurl}


\bibitem[Mitchell et~al\mbox{.}(2013)]%
        {mitchell2013pilaf}
\bibfield{author}{\bibinfo{person}{Christopher Mitchell}, \bibinfo{person}{Yifeng Geng}, {and} \bibinfo{person}{Jinyang Li}.} \bibinfo{year}{2013}\natexlab{}.
\newblock \showarticletitle{Using One-Sided {RDMA} Reads to Build a Fast, {CPU}-Efficient {Key-Value} Store}. In \bibinfo{booktitle}{\emph{Proceedings of the 2013 USENIX Annual Technical Conference}} (San Jose, CA, USA) \emph{(\bibinfo{series}{USENIX ATC '13})}. \bibinfo{publisher}{USENIX Association}, \bibinfo{address}{Berkeley, CA, USA}, \bibinfo{pages}{103–114}.
\newblock
\urldef\tempurl%
\url{https://www.usenix.org/conference/atc13/technical-sessions/presentation/mitchell}
\showURL{%
\tempurl}


\bibitem[Moir and Anderson(1995)]%
        {moir1995splitter}
\bibfield{author}{\bibinfo{person}{Mark Moir} {and} \bibinfo{person}{James~H. Anderson}.} \bibinfo{year}{1995}\natexlab{}.
\newblock \showarticletitle{Wait-free algorithms for fast, long-lived renaming}.
\newblock \bibinfo{journal}{\emph{Science of Computer Programming}} \bibinfo{volume}{25}, \bibinfo{number}{1} (\bibinfo{year}{1995}), \bibinfo{pages}{1--39}.
\newblock
\showISSN{0167-6423}
\urldef\tempurl%
\url{https://doi.org/10.1016/0167-6423(95)00009-H}
\showDOI{\tempurl}


\bibitem[Moraru et~al\mbox{.}(2013)]%
        {moraru2013epaxos}
\bibfield{author}{\bibinfo{person}{Iulian Moraru}, \bibinfo{person}{David~G. Andersen}, {and} \bibinfo{person}{Michael Kaminsky}.} \bibinfo{year}{2013}\natexlab{}.
\newblock \showarticletitle{There is More Consensus in Egalitarian Parliaments}. In \bibinfo{booktitle}{\emph{Proceedings of the 24th ACM Symposium on Operating Systems Principles}} (Farminton, PA, USA) \emph{(\bibinfo{series}{SOSP '13})}. \bibinfo{publisher}{Association for Computing Machinery}, \bibinfo{address}{New York, NY, USA}, \bibinfo{pages}{358–372}.
\newblock
\showISBNx{9781450323888}
\urldef\tempurl%
\url{https://doi.org/10.1145/2517349.2517350}
\showDOI{\tempurl}


\bibitem[Shen et~al\mbox{.}(2023)]%
        {fusee}
\bibfield{author}{\bibinfo{person}{Jiacheng Shen}, \bibinfo{person}{Pengfei Zuo}, \bibinfo{person}{Xuchuan Luo}, \bibinfo{person}{Tianyi Yang}, \bibinfo{person}{Yuxin Su}, \bibinfo{person}{Yangfan Zhou}, {and} \bibinfo{person}{Michael~R. Lyu}.} \bibinfo{year}{2023}\natexlab{}.
\newblock \showarticletitle{{FUSEE}: A Fully Memory-Disaggregated Key-Value Store}. In \bibinfo{booktitle}{\emph{Proceedings of the 21st USENIX Conference on File and Storage Technologies}} (Santa Clara, CA, USA) \emph{(\bibinfo{series}{FAST '23})}. \bibinfo{publisher}{USENIX Association}, \bibinfo{address}{Berkeley, CA, USA}, \bibinfo{pages}{81--98}.
\newblock
\showISBNx{978-1-939133-32-8}
\urldef\tempurl%
\url{https://www.usenix.org/conference/fast23/presentation/shen}
\showURL{%
\tempurl}


\bibitem[Technologies(2015)]%
        {rdma-manual}
\bibfield{author}{\bibinfo{person}{Mellanox Technologies}.} \bibinfo{year}{2015}\natexlab{}.
\newblock \bibinfo{title}{{RDMA} Aware Networks Programming User Manual. Rev 1.7}.
\newblock
\newblock
\urldef\tempurl%
\url{https://docs.nvidia.com/networking/display/rdmaawareprogrammingv17}
\showURL{%
\tempurl}
\newblock
\shownote{Accessed 2024-03-17}.


\bibitem[Tsai et~al\mbox{.}(2020)]%
        {tsai2020clover}
\bibfield{author}{\bibinfo{person}{Shin-Yeh Tsai}, \bibinfo{person}{Yizhou Shan}, {and} \bibinfo{person}{Yiying Zhang}.} \bibinfo{year}{2020}\natexlab{}.
\newblock \showarticletitle{Disaggregating Persistent Memory and Controlling Them Remotely: An Exploration of Passive Disaggregated Key-Value Stores}. In \bibinfo{booktitle}{\emph{Proceedings of the 2020 USENIX Annual Technical Conference}} (Virtual Event) \emph{(\bibinfo{series}{USENIX ATC '20})}. \bibinfo{publisher}{USENIX Association}, \bibinfo{address}{Berkeley, CA, USA}, Article \bibinfo{articleno}{3}, \bibinfo{numpages}{16}~pages.
\newblock
\showISBNx{978-1-939133-14-4}
\urldef\tempurl%
\url{https://www.usenix.org/conference/atc20/presentation/tsai}
\showURL{%
\tempurl}


\bibitem[Tseng et~al\mbox{.}(2023)]%
        {tseng2023distributed}
\bibfield{author}{\bibinfo{person}{Lewis Tseng}, \bibinfo{person}{Neo Zhou}, \bibinfo{person}{Cole Dumas}, \bibinfo{person}{Tigran Bantikyan}, {and} \bibinfo{person}{Roberto Palmieri}.} \bibinfo{year}{2023}\natexlab{}.
\newblock \showarticletitle{Distributed Multi-writer Multi-reader Atomic Register with Optimistically Fast Read and Write}. In \bibinfo{booktitle}{\emph{Proceedings of the 35th {ACM} Symposium on Parallelism in Algorithms and Architectures}} (Orlando, FL, USA) \emph{(\bibinfo{series}{SPAA '23})}. \bibinfo{publisher}{Association for Computing Machinery}, \bibinfo{address}{New York, NY, USA}, \bibinfo{pages}{479–488}.
\newblock
\showISBNx{9781450395458}
\urldef\tempurl%
\url{https://doi.org/10.1145/3558481.3591086}
\showDOI{\tempurl}


\bibitem[Wang et~al\mbox{.}(2022)]%
        {sherman2022}
\bibfield{author}{\bibinfo{person}{Qing Wang}, \bibinfo{person}{Youyou Lu}, {and} \bibinfo{person}{Jiwu Shu}.} \bibinfo{year}{2022}\natexlab{}.
\newblock \showarticletitle{Sherman: A Write-Optimized Distributed {B+}Tree Index on Disaggregated Memory}. In \bibinfo{booktitle}{\emph{Proceedings of the 2022 International Conference on Management of Data}} (Philadelphia, PA, USA) \emph{(\bibinfo{series}{SIGMOD '22})}. \bibinfo{publisher}{Association for Computing Machinery}, \bibinfo{address}{New York, NY, USA}, \bibinfo{pages}{1033–1048}.
\newblock
\showISBNx{9781450392495}
\urldef\tempurl%
\url{https://doi.org/10.1145/3514221.3517824}
\showDOI{\tempurl}


\bibitem[Zhou et~al\mbox{.}(2022)]%
        {zhou2022carbink}
\bibfield{author}{\bibinfo{person}{Yang Zhou}, \bibinfo{person}{Hassan M.~G. Wassel}, \bibinfo{person}{Sihang Liu}, \bibinfo{person}{Jiaqi Gao}, \bibinfo{person}{James Mickens}, \bibinfo{person}{Minlan Yu}, \bibinfo{person}{Chris Kennelly}, \bibinfo{person}{Paul Turner}, \bibinfo{person}{David~E. Culler}, \bibinfo{person}{Henry~M. Levy}, {and} \bibinfo{person}{Amin Vahdat}.} \bibinfo{year}{2022}\natexlab{}.
\newblock \showarticletitle{Carbink: {Fault-Tolerant} Far Memory}. In \bibinfo{booktitle}{\emph{Proceedings of the 16th USENIX Symposium on Operating Systems Design and Implementation}} (Carlsbad, CA, USA) \emph{(\bibinfo{series}{OSDI '22})}. \bibinfo{publisher}{USENIX Association}, \bibinfo{address}{Berkeley, CA, USA}, \bibinfo{pages}{55--71}.
\newblock
\showISBNx{978-1-939133-28-1}
\urldef\tempurl%
\url{https://www.usenix.org/conference/osdi22/presentation/zhou-yang}
\showURL{%
\tempurl}


\end{thebibliography}

\newpage\clearpage
\appendix
\section*{NON-PEER-REVIEWED APPENDICES}
\section{Reliable \wmr}
\label{app:wmr}

\begin{lstlisting}[caption={Reliable \wmr Pseudocode},label={alg:wmr}]
maxregs: shared array of unreliable max registers, initialized to [@$\bot$,\ldots,$\bot$@]
cache: local array of values, same length a@$ $@s maxregs, initialized to [@$\bot$,\ldots,$\bot$@]

def inner_write(v):
  for mr in 0..maxregs.length:@\label{line:write-for-start}@
    if cache[mr] < v:@\label{line:inner-write-if}@
      async: 
        maxregs[mr].write(v)@\label{line:update-mr}@
        cache[mr] = v@\label{line:update-cache-w}@ 
  wait until majority of cache >= v@\label{line:inner-write-wait}@
  
def WRITE(v):
  inner_write(v)

def READ():
  for mr in 0..maxregs.length:@\label{line:read-for-start}@
    async: cache[mr] = maxregs[mr].read() @\label{line:update-cache-r}@
  wait for a majority of answers@\label{line:read-wait}@
  v = max(cache.values)@\label{line:cache-max}@
  inner_write(v)@\label{line:invoke-inner-write}@
  return v
\end{lstlisting}

\paragraph{Interface} 
\wmr{}s have the following interface:
\begin{itemize}
  \item $\readop()$ $\rightarrow v$
  \item $\writeop(v) \rightarrow \t{ok}$ 
\end{itemize}

\paragraph{Properties}
\wmr{}s satisfy the following properties:
\begin{itemize}
  \item \textit{Validity}: If \readop $R$ returns $v$, then either (a) $v = \bot$, or (b) some operation $\writeop(v)$ was invoked before $R$ returns.
  \item \textit{Read-read monotonicity}: If a \readop{} returns value $y$ and a preceding \readop{} returns value $x$, then $x \leq y$.
  \item \textit{Write-read monotonicity}: If a \readop{} returns value $y$ and a preceding \writeop{} writes value $x$, then $x \leq y$.
  \item \textit{Liveness (wait-freedom)}: Every invoked operation eventually returns.
\end{itemize}

\subsection{Correctness}
We prove that our implementation of a reliable
\wmr object, shown in \Cref{alg:wmr}, satisfies the properties above.

\begin{theorem}
  \Cref{alg:wmr} satisfies validity.
\end{theorem}
\begin{proof}
  We say that a value $v$ is valid at time $t$ if either (a) $v = \bot$, or (b) some operation $\writeop(v)$ was invoked before time $t$. Thus the validity property can be restated as: if \readop $R$ returns $v$ at time $t$, then $v$ is valid at $t$. Note that if $v$ is valid at $t$, then $v$ is valid at any $t' \geq t$.

  We prove validity by showing that at any time $t$, \t{maxregs} and \t{cache} only contain valid values. Then validity follows immediately, since the value returned by \t{READ} at $p$ is the maximum value in $p$'s cache.

  Initially \t{maxregs} and \t{cache} only contain $\bot$, which is a valid value. Suppose by contradiction that there is a time $t$ and a value $v$ such that \t{maxregs} or \t{cache} contains $v$ at time $t$ and $v$ is not valid at $t$. Assume \textit{wlog} that $t$ is the earliest such time.

  Time $t$ must correspond to a step by some process $p$ which updates either \t{maxregs} or \t{cache}. Let us examine all such possible steps:
  \begin{itemize}
  \item \Cref{line:update-mr}. The value $v$ is the argument of the \t{inner\_write} function, thus either $v$ is the argument of a \writeop at time $t' \leq t$, or $v$ is obtained at \Cref{line:cache-max} at time $t'' < t$. In the former case, $v$ must be valid at time $t'$ and thus also at time $t$, since it is the input of a \writeop{} invoked before $t'$. In the latter case, $v$ is valid at time $t''$, since $t'' < t$ and $t$ is the earliest time when \t{maxregs} or \t{cache} contain an invalid value. In both cases, $v$ is valid at $t$, a contradiction.
  \item \Cref{line:update-cache-w}. By the same argument as above, $v$ must be valid at $t$, a contradiction.
  \item \Cref{line:update-cache-r}. This line contains two steps: first, at time $t' < t$, some value $v$ is read from \t{maxregs[mr]}, and then, at time $t$, $v$ is written into \t{cache[mr]}. By our assumption $v$ must be valid at time $t'$ (thus also at $t$), so we again reach a contradiction. 
  \end{itemize}
  All possibilities lead to a contradiction, thus it must be the case that \t{maxregs} and \t{cache} only contain valid values at all times. This completes the proof.
\end{proof}

\begin{observation}\label{obs:maj-higher}
  If \t{inner\_write(v)} completes at time $t$ at process $p$, then a majority of $p$'s cache slots contain values $\geq v$ at time $t$.
\end{observation}

\begin{lemma}\label{lem:cache-maxreg}
  At any process $p$ and time $t$, for any $0 \leq mr \leq \t{maxregs.length}$, if $\t{cache[mr]} = v$ at time $t$, then some \t{maxregs[mr].write(v)} or \t{maxregs[mr].read() $\rightarrow$ v} completed before $t$. 
\end{lemma}

\begin{proof}
  Observe that \t{cache[mr]} is only updated at \Cref{line:update-cache-w} and \Cref{line:update-cache-r}. If \t{cache[mr]} becomes equal to $v$ at \Cref{line:update-cache-w}, then \t{maxregs[mr].write(v)} was called in the previous line. Otherwise, if \t{cache[mr]} becomes equal to $v$ at \Cref{line:update-cache-r}, then \t{maxregs[mr].read()} returned $v$ on the same line.
\end{proof}

\begin{lemma}\label{lem:read-after-inner-write}
  If \t{inner\_write(v)} completes at time $t$, then after $t$, reading from all \t{maxregs} and waiting for a majority of responses produces at least one value $\geq v$.
\end{lemma}
\begin{proof}
  Let $p$ be the process that called \t{inner\_write(v)}. 
  By \Cref{obs:maj-higher}, a majority $M_1$ of $p$'s cache slots contain values $\geq v$ at $t$. 
  Let $M_2$ be the majority of \t{maxregs} for which \readop responses are received. $M_1$ and $M_2$ must intersect in at least one index $i$. Since $i\in M_1$, by \Cref{lem:cache-maxreg}, some \t{maxregs[$i$].write(v')} or \t{maxregs[$i$].read() $\rightarrow$ v'} completed before $t$, where $v' \geq v$. Therefore, by the read-read or write-read monotonicity of \t{maxregs[$i$]}, the \readop{} of \t{maxregs[$i$]} after $t$ returns a value $\geq v$, as required by the lemma.
\end{proof} 

\begin{theorem}
  \Cref{alg:wmr} satisfies read-read monotonicity.
\end{theorem}
\begin{proof}
  Let $R_1$ and $R_2$ be two \readop{}s such that $R_1$ precedes $R_2$, $R_1$ returns $x$, and $R_2$ returns $y$. We will show that $x \geq y$.

  By \Cref{lem:read-after-inner-write}, since $R_1$ calls \t{inner\_write($x$)} before it returns, at least one of the \readop{}s performed by $R_2$ in lines~\ref{line:read-for-start}--\ref{line:read-wait} returns a value $\geq x$. So the cache of the invoking process of $R_2$ contains at least one value $\geq x$ at the time when it invokes \Cref{line:cache-max}. Thus the return value of $R_2$ is $\geq x$.
  
  \end{proof}

\begin{theorem}
  \Cref{alg:wmr} satisfies write-read monotonicity.
\end{theorem}
\begin{proof}
  Let $W$ be a \writeop{} and $R$ be a \readop{} such that $W$ precedes $R$, $W$ writes $x$, and $R$ returns $y$. We will show that $x \geq y$.

  By \Cref{lem:read-after-inner-write}, since $W$ calls \t{inner\_write($x$)} before it returns, at least one of the \readop{}s performed by $R$ in lines~\ref{line:read-for-start}--\ref{line:read-wait} returns a value $\geq x$. So the cache of the invoking process of $R$ contains at least one value $\geq x$ at the time when it invokes \Cref{line:cache-max}. Thus the return value of $R$ is $\geq x$. 

  \end{proof}

\begin{theorem}
  \Cref{alg:wmr} satisfies wait-freedom.
\end{theorem}
\begin{proof}
  We start by showing that the \t{inner\_write} function is wait-free. Let $p$ be the process that invokes it. The only blocking step of the function is \Cref{line:inner-write-wait}. If $p$'s local cache already contains a majority of values $\geq v$, then the function returns immediately, without waiting for any asynchronous operation invoked in the for loop. Otherwise, the local cache contains a majority of values $<v$, so the if statement at \Cref{line:inner-write-if} is triggered for a majority of indices. Thus, $v$ will be written asynchronously to a majority of locations in $p$'s cache. At the latest when all such \writeop{}s have completed, the wait statement at \Cref{line:inner-write-wait} and the function will return.

  The \writeop{} operation is a mere call to \t{inner\_write}, so it is clearly wait-free as well.

  The \readop{} operation has an additional blocking step at \Cref{line:read-wait}. This step completes once a majority of the max registers have been read from, which cannot be blocked or prevented. Then, the \readop{} operation invokes \t{inner\_write}, which is wait-free and thus must eventually return. So the \readop{} operation is wait-free as well.
\end{proof} 

\subsection{Roundtrip Complexity}\label{sec:wmr-complexity}

We assume that a \writeop to, or \readop from, one of the underlying unreliable max registers takes 1 roundtrip time (RTT) (this is true in our implementation). Then, the \t{inner\_write} function returns in at most 1 RTT, since all of the \writeop{}s in \Cref{line:update-mr} are performed asynchronously in parallel. Furthermore, if the \t{cache} already contains $v$, or a larger value, in a majority of slots, then \t{inner\_write} returns immediately, in 0 RTTs, even if the if statement at \Cref{line:inner-write-if} is triggered for a minority of cache slots. This is because the resulting \writeop{}s to \t{maxregs} are performed in the background and do not block \t{inner\_write} from returning. Since the \t{WRITE} function is just a wrapper around \t{inner\_write}, all of the above is also true about \t{WRITE}.

The \t{READ} function always requires 1 RTT to complete lines~\ref{line:read-for-start}--\ref{line:read-wait}, after which it calls the \t{inner\_write} function which, as discussed above, may incur 0 or 1 RTTs. Overall, \t{READ} incurs 1 or 2 RTTs. However, there are two common-case scenarios in which \t{READ} requires only 1 RTT:
\begin{itemize}
  \item Under low contention, we can assume that \writeop{}s are not concurrent with any other operation, and that enough time passes between a \t{WRITE(v)} and a subsequent \t{READ} such that the background \writeop{}s of \t{v} to \t{maxregs} complete before the \readop{} begins. 
  Then, when the \t{READ} performs the \readop{}s at lines~\ref{line:read-for-start}--\ref{line:read-wait}, all such \readop{}s will return \t{v}, and thus will populate a majority of the local cache with $v$. Therefore, the \t{inner\_write} at \Cref{line:invoke-inner-write} will return in 0 RTTs.
  \item It is often the case that, among a collection of machines, some are consistently more responsive than others. If this is the case for the memories hosting each max register in \t{maxregs}, then all operations will receive responses first from the same majority of \t{maxregs} for writing and reading (lines~\ref{line:write-for-start}--\ref{line:inner-write-wait} and lines~\ref{line:read-for-start}--\ref{line:read-wait}). If this occurs, and there is low contention (\writeop{}s are not concurrent with any other operation), then \readop{}s always return in 1 RTT, even if operations occur one right after the other, without any background operations completing in the meantime.
\end{itemize} 

In ABD and \protocol, when a \wmr is read to obtain a fresh timestamp as part of a \writeop{} (Algorithms~\ref{alg:abd} and \ref{alg:sg-write}, respectively), read-read monotonicity is unnecessary.
It is then enough to use a weaker \readop{} operation that does not call \t{inner\_write(v)} and always completes in 1 RTT. Our proof does not cover this optimization.
\section{Timestamp Lock}\label{app:timestamp}

\begin{lstlisting}[caption={Timestamp Lock Pseudocode},label={app:split}]
CASes = {(@$\bot$@, @$\bot$@), ...} // 2f+1 CAS Objects @\label{app:split:CASs}@

def TRYLOCK(ts, mode: READ or WRITE):
  read: dict<CAS, CasValue> = {(@$\bot$@, @$\bot$@), ...}
  async for c in CASes: @\label{app:split:async}@
    while read[c].ts < ts: @\label{app:split:while}@
      expected = read[c]
      read[c] = c.CAS(expected, (ts, mode))@\label{app:split:cas}@
      if read[c] == expected: break@\label{app:split:break}@
  wai@$ $@t fo@$ $@r a majority to complete @\label{app:split:wait}@
  if any c st read[c].ts > ts:@\label{app:split:any-greater}@ return False @\label{app:split:return-1}@
  if any c st read[c] == (ts, @$\neg$@mode):@\label{app:split:any-other-side}@ return False@\label{app:split:return-2}@
  return True @\label{app:split:return-3}@
\end{lstlisting}

\paragraph{Interface}
Timestamp locks have a single operation:

\begin{itemize}
  \item \t{TRYLOCK(ts, mode: READ or WRITE)} $\rightarrow$ \t{bool} 
\end{itemize}

\paragraph{Properties}

Timestamp locks satisfy three properties:
\begin{itemize}
  \item \true safety: if a lock operation \t{TRYLOCK(ts,mode)} returns $r$ at time $t$, and no operation \t{TRYLOCK(ts',mode')} was invoked before $t$ such that (i) $ts' > ts$, or (ii) $ts'=ts$ and $mode' = \neg mode$, then $r = \true$. 
  
  \item \true exclusion: \t{TRYLOCK(ts,m)} and \t{TRYLOCK(ts,$\neg$m)} cannot both return \true.
  
  \item Liveness (wait-freedom): Every operation eventually returns.
\end{itemize}

\subsection{Correctness}
\begin{theorem}
  \Cref{app:split} satisfies \true safety.
\end{theorem}
\begin{proof}
  Let us define a \t{TRYLOCK(ts',mode')} operation to be \textit{$(ts,mode,t)$-conflicting} if it is invoked before time $t$ and either (i) $ts'>ts$ or (ii) $ts'=ts$ and $mode' = \neg mode$. Thus, the theorem can be stated as: if a lock operation \t{TRYLOCK(ts,mode)} returns $r$ at time $t$, and no $(ts,mode,t)$-conflicting operation exists, then $r = \true$.

  Let $D$ be a \t{TRYLOCK(ts,m)} operation that returns at time $t$ and assume no $(ts,mode,t)$-conflicting operation exists. We will show that $D$ must return \true. Since no $(ts,mode,t)$-conflicting operation exists, by time $t$ none of the shared CAS objects contains a tuple $(ts',mode')$ such that $ts'>ts$, or $ts'=ts$ and $mode' = \neg mode$, since a CAS object is populated with $(ts,mode)$ only at \Cref{app:split:cas} and only if $(ts,mode)$ are the arguments to \t{TRYLOCK}. Thus also none of the fields of the \t{read} object in $D$ can contain such a tuple either by time $t$, since this \t{read} array is only populated with values from the CAS objects (\Cref{app:split:cas}). It then follows that neither of the if statements at \Cref{app:split:any-greater} or \Cref{app:split:any-other-side} is triggered, so $D$ must return \true at \Cref{app:split:return-3}.
\end{proof}

\begin{observation}\label{obs:cas-increasing}
  In \Cref{app:split},
  CAS objects
  never decrease in value. Furthermore,
  CAS objects
  never change their contents from $(ts,mode)$ to $(ts,\neg{mode})$ for any $ts,mode$.
\end{observation}
\begin{proof}
  CAS objects
  are only
  updated at \Cref{app:split:cas}.
  Take a CAS object $c$.
  At \Cref{app:split:cas}, \t{c.CAS(expected, (ts, mode))} is successful only if (1) $c$ already contains \t{expected}, which is equal to \t{read[c]}, and (2) $ts > \t{read[c]}$. Thus, $c$ takes strictly monotonically increasing values over time.
\end{proof}

\begin{theorem}
  \Cref{app:split} satisfies \true exclusion.
\end{theorem}
\begin{proof}
  Assume by contradiction that there exist two operations $D_1 = \t{TRYLOCK(ts,mode)}$ and $D_2 = \t{TRYLOCK(ts,$\neg$mode)}$ such that $D_1$ and $D_2$ both return \true.

  Since $D_1$ returns, a majority of iterations of the for loop at \Cref{app:split:async} must complete. Denote by $M_1$ the set of
  CAS objects
  whose iterations complete. Similarly for $D_2$; denote by $M_2$ the set of
  CAS objects
  whose iterations complete at $D_2$. $M_1$ and $M_2$ must intersect in at least one
  CAS object $m$.
  $D_1$ must observe $m$ to contain $(ts, mode)$, otherwise $D_1$ would return \true at \Cref{app:split:return-1} or \Cref{app:split:return-2}. Similarly, $D_2$ must observe $m$ to contain $(ts,\neg mode)$. Thus, $m$ must change in value from $(ts,\neg mode)$ to $(ts, mode)$ or vice-versa, possibly passing through one or more intermediate values. By \Cref{obs:cas-increasing}, $m$ cannot pass directly from $(ts, mode)$ to $(ts,\neg mode)$ or vice-versa. Furthermore, also by \Cref{obs:cas-increasing}, if $m$ changes from $(ts,\cdot)$ to some other value $(ts',\cdot)$, then $ts'>ts$, after which $m$ will never again contain $(ts,\cdot)$. We have reached a contradiction in all cases, thus the theorem must hold.
\end{proof}

\begin{theorem}
  \Cref{app:split} satisfies wait-freedom.
\end{theorem}

\begin{proof}
  The only blocking step of the \t{TRYLOCK} operation is at \Cref{app:split:wait}. To prove that this step eventually completes, we show that a majority of the iterations of the for loop eventually complete. 

  Let $c$ be a correct CAS object,
  i.e., one that never fails. We show that at each iteration of the while loop, either (a) the CAS operation succeeds, or (b) the contents of \t{read[c]} strictly increase. Assume the CAS operation fails and consider two cases:
  \begin{itemize}
  \item \t{read[c] = $\bot$}. Since the CAS fails, \t{read[c]} becomes
  the current value of $c$, which cannot be equal to $\bot$, otherwise the CAS would succeed. 
  Since any value is greater than $\bot$, \t{read[c]} strictly increases in value.

  \item \t{read[c] $\ne\bot$}. In this case, \t{read[c]} must be equal to the contents of $c$ from a previous iteration of the while loop. After the CAS operation, \t{read[c]} becomes equal to the current contents of $c$. By \Cref{obs:cas-increasing}, the contents of $c$ can only strictly increase in value, thus \t{read[c]} also strictly increases in value in this case.
  \end{itemize}
  
  We have shown that at each iteration of the while loop, either (a) the CAS operation succeeds, or (b) the contents of \t{read[c]} strictly increase. In case (a), the while loop will exit immediately afterwards, due to the check at \Cref{app:split:break}. In case (b), there can only be a finite number of such iterations, since $ts$ is a fixed value and thus \t{read[c]} can only take a finite number of values before it becomes greater than $ts$, at which point the while loop, and corresponding for loop iteration, complete.

  We have shown that for each correct CAS object, the corresponding for loop iteration eventually completes. Since we assume there are a majority of such correct CAS objects in the system, it follows that a majority of the iterations complete. Thus, each \t{TRYLOCK} operation eventually completes.
\end{proof}

\subsection{Roundtrip Complexity}
We assume that performing a CAS on an underlying CAS object incurs a single RTT. Thus, each iteration of the while loop at \Cref{app:split:while} takes 1 RTT. Therefore, the round complexity of a \t{TRYLOCK} call is equal to the highest number of iterations of the while loop performed before a majority of iterations of the for loop have completed. 

Fix a CAS object \t{c}. At the start of the \t{TRYLOCK} call, \t{read[c] = $\bot$}. In the worst case, \t{read[c]} takes a new value between $\bot$ and \t{ts} every time the CAS at \Cref{app:split:cas} is attempted. Since timestamps can only take positive, integer values, this means that at most $\t{ts} + 1$ iterations can be performed before \t{c} takes a value $\geq \t{ts}$. Therefore, the worst case complexity of \t{TRYLOCK(ts, mode)} is $\t{ts} + 1$. 
\section{\protocol}\label{app:safeguess}

\begin{lstlisting}[caption={\protocol{}' Pseudocode},label={app:alg:sg},float]
M = ((0, @$\bot$@), VERIFIED, @$\bot$@) // Max Register@\label{app:alg:sg-write:M}@
TSL[tid] = {} // Timestamp Lock @\label{app:alg:sg-write:S}@

def WRITE(v):
  w = (guessTs(), GUESSED, v) @\label{app:line:sg-write:guess-ts}@
  i@$ $@n parallel {m = M.READ(), M.WRITE(x)} @\label{app:line:sg-write:parallel}@
  if m <= w: // Fast path @\label{app:line:sg-write:fp}@
    i@$ $@n bg: M.WRITE(w with VERIFIED) // Spdup reads @\label{app:line:sg-write:bg-write}@
  else: // Slow path (potentially stale timestamp) @\label{app:line:sg-write:sp}@
    if TSL[tid].TRYLOCK(w.ts, WRITE): @\label{app:line:sg-write:split}@
      M.WRITE(((m.ts.i+1,tid), VERIFIED, v)) @\label{app:line:sg-write:verified}@

def READ():
  seen: dict<ThreadId, MValue> = {} @\label{app:line:sg-read:dict}@
  while True: @\label{app:line:sg-read:loop}@
    m = M.READ() @\label{app:line:sg-read:read}@
    if m is VERIFIED: return m.v // Fast path @\label{app:line:sg-read:fp}@
    if m in seen.values: // Fresh timestamp @\label{app:line:sg-read:fresh}@
      if TSL[m.ts.tid].TRYLOCK(m.ts, READ):@\label{app:line:sg-read:split}@
        i@$ $@n bg: M.WRITE(m with VERIFIED) // Spdup rds @\label{app:line:sg-read:write}@
        return m.v @\label{app:line:sg-read:return-1}@
    elif m.ts.tid in seen.keys: // Wait-free path@\label{app:line:sg-read:in-seen}@
      return seen[m.ts.tid].v @\label{app:line:sg-read:return-2}@
    seen[m.ts.tid] = m @\label{app:line:sg-read:update-seen}@
\end{lstlisting}

\protocol implements a wait-free atomic multi-writer multi-reader register, with the standard read/write interface and usual correctness properties: linearizability and wait-freedom.

\subsection{Linearizability}
We prove the linearizability of \protocol by describing an explicit linearization for each execution of \protocol and then showing that the linearization is valid: (1) the linearization is a legal sequential execution of a register, (2) the linearization is equivalent to the concurrent execution from which it is derived, and (3) the linearization preserves the real-time ordering of the concurrent execution.

\subsubsection{Preliminaries}

We define the notion of \textit{timestamp of a \writeop}. 
\begin{definition}\label{def:write-timestamp}
  We associate each \writeop $W$ with a timestamp $TS(W)$, in the following way:
  \begin{itemize}
  \item If $W$ performs the \writeop at \Cref{app:line:sg-write:verified}, then $TS(W)$ is the timestamp computed by $W$ in \Cref{app:line:sg-write:verified}. We also call this the verified timestamp of $W$.
  \item Otherwise, then $TS(W)$ is the timestamp guessed by $W$ in \Cref{app:line:sg-write:guess-ts}. We also call this the guessed timestamp of $W$.
\end{itemize} 
\end{definition}

We also rely on a number of observations, assumptions and lemmas in our proof of linearizability. We start with the following observation, which follows from the fact that \readop{s} only return values read from \t{M}, and \writeop{}s to \t{M} are always input values of \writeop{s}. 
\begin{observation}\label{obs:read-from-write}
  In \Cref{app:alg:sg}, the return value of a \readop{} operation is the input value of a \writeop{} operation, or the initial value $\bot$.
\end{observation}

Furthermore, we make the (reasonable) assumption that a process guesses for each \writeop{} a higher timestamp than it used for all previous \writeop{}s. If the verified timestamp is computed, it is even higher, thus higher than all previous timestamps used by this process.
\begin{assumption}\label{ass:monotone-ts}
  In \Cref{app:alg:sg}, if $W_1$ and $W_2$ are two \writeop{}s by the same process such that $W_1$ precedes $W_2$, then $TS(W_1) < TS(W_2)$.
\end{assumption}

We also make the standard assumption~\cite{herlihy1990-linearizability} that each thread has at most one outstanding operation at a time.
\begin{assumption}\label{ass:one-outstanding}
  In any valid execution, a thread $T$ invokes a new operation only after $T$'s current operation completes.
\end{assumption}

The next observation follows from Assumptions~\ref{ass:monotone-ts} and \ref{ass:one-outstanding}, and from the fact that timestamps include the process identifier.
\begin{observation}
  In \Cref{app:alg:sg}, no two \writeop{}s have the same timestamp, in the sense of \Cref{def:write-timestamp}.
\end{observation}

We next prove the following technical lemma about our reliable max register, which will help us reason about repeated \readop{s} of \t{M} performed in different iterations of the \t{READ} function's while loop.
\begin{lemma}[Double-read lemma]\label{lem:double-read}
  Let $M$ be a Reliable \wmr, and $a < b$ two values. Let $R_1$ and $R_2$ be two \readop{}s from $M$ such that $R_1$ precedes $R_2$ and both $R_1$ and $R_2$ return $a$. Then, no \writeop{} of $b$ precedes the first invocation of a \writeop{} of $a$.
\end{lemma}
\begin{proof}
  Assume not. Let $W(a)$ be the \writeop{} of $a$ with the earliest invocation. Let $W(b)$ be a \writeop{} of $b$ such that $W(b)$ precedes $W(a)$. We have the following:
  \begin{itemize}
  \item $t_1 > t_2$, where $t_1$ is the return time of $W(b)$ and $t_2$ is the invocation time of $R_2$. This is because of the write-read monotonicity of $M$.
  \item $t_2 > t_3$, where $t_3$ is the return time of $R_1$. This is by our assumption that $R_1$ precedes $R_2$.
  \item $t_3 > t_4$, where $t_4$ is the invocation time of $W(a)$. This is because of the validity property of $M$.
  \item $t_4 > t_1$, by our assumption that $W(b)$ precedes $W(a)$.
  \end{itemize}
  We have reached a circular inequality, thus the lemma must hold.
\end{proof}

Finally, we prove the following lemma which helps explain what happens if \t{TRYLOCK} at \Cref{app:line:sg-write:split} returns \false.
\begin{lemma}\label{lem:false-propose-read}
  In \Cref{app:alg:sg}, let $W$ be a \writeop{} operation and \t{w} the tuple it initializes at \Cref{app:line:sg-write:guess-ts}. If $W$ performs the \t{TRYLOCK} at \Cref{app:line:sg-write:split} and gets \false, then some \readop{} operation reads \t{w} from \t{M} in two separate iterations of its while loop.
\end{lemma}
\begin{proof}
  Let $p$ be the process that executes $W$. By the \true safety property of the timestamp lock, if $W$'s trylock call returns \false, then there exists some concurrent or preceding trylock call $A$ with either (1) a higher timestamp than \t{w}, or (2) the same timestamp but in read mode. Case (1) is impossible: if $A$ is invoked by a \writeop{}, this \writeop{} must be performed by $p$ (since timestamp locks are indexed per process in \Cref{app:alg:sg}) earlier than $W$ with a higher timestamp---impossible by \Cref{ass:monotone-ts}; if $A$ is invoked by a \readop{}, this \readop{} must have seen the higher timestamp associated to a value written by $p$---impossible. Thus, case (2) must be true: $A$ must have the same timestamp as $W$ but used the read lock mode. This is only possible if a \readop{} $R$ has read the value of $W$ from \t{M}. Furthermore, $R$ must have read this value twice, since $R$ invokes \t{TRYLOCK} only if $R$ finds the current value in its \t{seen} collection, meaning that the same value was read in a previous iteration.
\end{proof}

\subsubsection{Linearization Rules \& Proof}
Now we can give our linearization construction. Given an execution $\mathcal{E}$ of \protocol, we linearize it as follows:
\begin{enumerate}
  \item Linearize \writeop{}s in order of increasing timestamps, in the sense of \Cref{def:write-timestamp}.
  \item If a \readop{} $R$ returns $\bot$, linearize $R$ before the first \writeop{}.
  \item Otherwise, linearize a \readop{} $R$ after the \writeop{} whose value it returns (\Cref{obs:read-from-write}), but before the following \writeop{}. 
  \item If two \readop{}s $R_1$ and $R_2$ return the same value and $R_1$ precedes $R_2$ in $\mathcal{E}$, linearize $R_1$ before $R_2$.
\end{enumerate}

We denote by $\mathcal{L(E)}$ the resulting linearization. We now show that $\mathcal{L(E)}$ is valid. Clearly, $\mathcal{L(E)}$ is a legal sequential execution of a register, since every \readop{} returns the value of the latest \writeop{}, with the exception of \readop{}s occurring before the first \writeop{}, which return $\bot$. Furthermore, $\mathcal{L(E)}$ is equivalent to $\mathcal{E}$, since all operations return the same values.

It only remains to prove that $\mathcal{L(E)}$ preserves the real-time ordering present in $\mathcal{E}$. We show that if some operation $O_1$ precedes another operation $O_2$ in $\mathcal{E}$, then the same is true in $\mathcal{L(E)}$. We consider all four cases $(O_1, O_2) \in \{\text{write}, \text{read}\}\times\{\text{write}, \text{read}\}$.

\subsubsection*{$(O_1, O_2) = (\text{write}, \text{write})$}
We rely on the following lemma:
\begin{lemma}
  If $W_1$ and $W_2$ are two \writeop{}s such that $W_1$ precedes $W_2$, then $TS(W_1) < TS(W_2)$.
\end{lemma}
\begin{proof}
  By the write-read monotonicity of \t{M}, the \readop{} of $W_2$ returns a value \t{m} with a timestamp not lower than $TS(W_1)$. We distinguish three cases:
  \begin{itemize}
  \item $W_2$ performs the fast path. In this case, $\t{m} \leq \t{w}$, where \t{w} is the value computed at \Cref{app:line:sg-write:guess-ts}. By \Cref{def:write-timestamp}, \t{w} is also equal to $TS(W_2)$. It follows that $TS(W_2) > TS(W_1)$.
  
  \item $W_2$ performs the \writeop{} at \Cref{app:line:sg-write:verified}. In this case, we have $TS(W_2) = \t{m.ts + 1} > \t{m.ts} \geq TS(W_1)$.
  
  \item $W_2$ enters the slow path, but does not perform the \writeop{} at \Cref{app:line:sg-write:verified}. In this case, the \t{TRYLOCK} call at \Cref{app:line:sg-write:split} must have returned \false. By \Cref{lem:false-propose-read}, some \readop{} must have read \t{w} twice, sequentially. By \Cref{lem:double-read}, no \writeop{} to \t{M} with a higher timestamp than \t{w} could have preceded $W_2$'s \writeop{} of \t{w} at \Cref{app:line:sg-write:parallel}. So $TS(W_1)$ must be lower than $\t{w.ts} = TS(W_2)$.
  \end{itemize}
\end{proof}
Thus, according to the first linearization rule, $W_2$ is linearized after $W_1$, as required.

\subsubsection*{$(O_1, O_2) = (\text{read}, \text{read})$}
Let $R_1$ and $R_2$ be two \readop{}s such that $R_1$ precedes $R_2$. We show that $R_1$ is linearized before $R_2$.

In case $R_1$ and $R_2$ return the same value, then by the fourth linearization rule, $R_1$ is linearized before $R_2$.

In case $R_1$ and $R_2$ return different values, let $v_1$ and $v_2$ be their respective return values. Let $r_1$ be the (latest) \t{M.READ} call inside $R_1$ that returns $v_1$ and $r_2$ be the (latest) \t{M.READ} call inside $R_2$ that returns $v_2$. Since $r_1$ precedes $r_2$, by the read-read monotonicity of \t{M}, it must be that $v_1 < v_2$. By the second and third linearization rules, $R_1$ is linearized before the \writeop{} of $v_2$, and $R_2$ is linearized after the \writeop{} of $v_2$, so $R_1$ is linearized before $R_2$.

\subsubsection*{$(O_1, O_2) = (\text{write}, \text{read})$}
Let $W$ and $R$ be a \writeop{} and a \readop{}, respectively, such that $W$ precedes $R$. We show that $W$ is linearized before $R$.

Let $v$ be the return value of $R$ and let $r$ be the earliest \t{M.READ} call inside $R$ that returns $v$. By the write-read monotonicity of \t{M}, the timestamp associated with $v$ cannot be lower than the timestamp of $W$. In other words, $R$ returns a value installed by $W$ or a later \writeop{}, so by the third linearization rule $R$ must be linearized after $W$.

\subsubsection*{$(O_1, O_2) = (\text{read}, \text{write})$} Let $R$ and $W$ be a \writeop{} and a \readop{}, respectively, such that $R$ precedes $W$. We show that $R$ is linearized before $W$.

Let $T$ be the timestamp associated to the value returned by $R$. Since $W$ is invoked after $R$
returns, by read-read monotonicity, $W$'s \readop{} at \Cref{app:line:sg-write:parallel} returns a value \t{m} with a timestamp
$T_m \geq T$.
We distinguish three cases:
\begin{itemize}
  \item $W$ enters the fast path. In this case, $W$'s
  guessed timestamp must have been greater than $T_m$, so in turn greater than $T$.
  
  \item $W$ performs the \writeop{} at \Cref{app:line:sg-write:verified}. In this case, $W$ picks a timestamp greater than $T_m$, so in turn greater than $T$.

  \item $W$ enters the fast path but does not perform the \writeop{} at \Cref{app:line:sg-write:verified}. Then, by \Cref{lem:false-propose-read}, some \readop{} operation must have read \t{w} (i.e., $W$'s value, initialized at \Cref{app:line:sg-write:guess-ts}) from \t{M} twice, sequentially. The second such \readop{} must have been invoked after $W$ was invoked, so by read-read monotonicity, the timestamp this second \readop{} sees (which is the timestamp of \t{w}) must be greater than $T$.
\end{itemize}

We have shown that in all cases, $W$ has a higher timestamp than the value returned by $R$. $R$ therefore cannot return $W$'s value, or a later one, and is linearized before $W$.

\subsection{Wait-Freedom}
The \t{WRITE} function is clearly wait-free, since it does not have any blocking steps, and the external functions it calls (\t{M.READ}, \t{M.WRITE}, \t{TSL.TRYLOCK}) are all wait-free.

The case of the \t{READ} function is less straightforward. We start with the following observation, which follows from the fact that the \t{WRITE} and \t{READ} functions call \t{TRYLOCK(m.ts, *)} only after having written \t{m} to \t{M} or read \t{m} from \t{M}, respectively.

\begin{observation}\label{obs:propose}
  In \Cref{app:alg:sg}, if \t{TSL[$p$].TRYLOCK(m.ts, *)} is invoked at time $t$, then either \t{M.WRITE(m)} or \t{M.READ() $\rightarrow$ m} completed before $t$. 
\end{observation}

Using \Cref{obs:propose}, we prove the following lemma, which crucially bounds the number of times a \readop{} can see values from each process before the \readop{} is forced to return. 

\begin{lemma}\label{lem:max-two-iters}
  If a \readop{} operation $R$ reads a value from the same process $p$ in more than two iterations of $R$'s while loop, then $R$ returns on the third iteration in which $R$ reads from $p$. 
\end{lemma}
\begin{proof}
  Assume not: $R$ reads a value from the same process $p$ in three iterations, without returning on the third such iteration. 
  Clearly, none of the values can be \t{VERIFIED}, otherwise $R$ would return. 
  Furthermore, since $R$ returns as soon as it sees two distinct values from the same process, the values must all be equal to each other: call this value \t{m}.
  The second time $R$ reads \t{m}, it invokes \t{TSL[m.ts.tid].TRYLOCK(m.ts, READ)} at \Cref{app:line:sg-read:split}. This call, denoted by $C$, must return \false, otherwise $R$ would return on the second iteration. 

  By \true safety, some \t{TSL[$p$].TRYLOCK(m'.ts, mode’)} call must be concurrent with, or precede, $C$, with either (i) $\t{m'.ts} > \t{m.ts}$, or (ii) $\t{m'.ts} = \t{m.ts}$ and $\t{mode'} = \t{WRITE}$. 
  \begin{itemize}
  \item In case (i), by \Cref{obs:propose}, some preceding \readop{} or \writeop{} of \t{m'} must exist. Thus, in the third iteration of $R$ which reads from $p$, by \t{M}'s monotonicity properties, \t{M.READ} cannot return \t{m} again.
  This means that in this third iteration that reads from $p$, $R$ sees a value different from \t{m}, and returns.
  Contradiction.
  
  \item In case (ii), $p$ must have called \t{TSL[$p$].TRYLOCK(m.ts, WRITE)} as part of the slow path of some \writeop{} $W$, so $W$'s \t{M.READ} at \Cref{app:line:sg-write:guess-ts} returned a value \t{m'} with a higher timestamp than that of \t{m}. By read-read monotonicity, in the third iteration of $R$ which reads from $p$, \t{M.READ} cannot return \t{m} again.
  This again means that in this third iteration that reads from $p$, $R$ sees a value different from \t{m}, and returns.
  Contradiction. 
  \end{itemize}
  We have reached a contradiction in all cases, thus the lemma must hold.
\end{proof}

Wait-freedom follows easily from \Cref{lem:max-two-iters}: since a \readop{} can see values from each process in at most two iterations without returning, each \readop{} must return after at most $2n + 1$ iterations, where $n$ is the number of processes.

\subsection{Roundtrip Complexity}
Recall that \t{M.WRITE} incurs 0 or 1 RTTs, \t{M.READ} incurs 1 or 2 RTTs, and \t{TRYLOCK(ts,p)} takes between 1 and $\t{ts} + 1$ RTTs.

First, let us examine the round complexity of \protocol's \t{WRITE} function. \Cref{app:line:sg-write:parallel} takes between 1 and 2 RTTs. If the fast path is entered, no further RTTs are incurred on the critical path. Otherwise, the slow path is entered,
and
\t{TRYLOCK} takes between 1 and \t{w.ts + 1} RTTs, where \t{w.ts} is the output of \t{guessTs()}. Finally, if \t{TRYLOCK} returns \true, then the \writeop{} at \Cref{app:line:sg-write:verified} takes an additional 0 or 1 RTTs. Overall, the \t{WRITE} function incurs 1 or 2 RTTs on the fast path and between 2 and $2 + \t{w.ts} + 1 + 1 = \t{w.ts} + 4$ RTTs on the slow path.

Now let us examine the round complexity of \protocol's \t{READ} function. Each iteration of the while loop incurs 1 or 2 RTTs for the \t{M.READ} at \Cref{app:line:sg-read:read}, as well as between 1 and $1 + \t{m.ts}$ RTTs in case the \t{TRYLOCK} at \Cref{app:line:sg-read:split} is invoked. So each iteration of the while loop incurs between 1 and
  $3 + \t{m.ts}$ RTTs.
  As argued above, a \readop{} performs at most $2n+1$ iterations of the while loop, so the overall complexity of a \readop{} is between 1 and
  $(2n+1)(3+\t{m.ts})$ RTTs.

\subsection*{When do \writeop{}s and \readop{}s return in a single RTT?}

A \writeop{} returns in a single RTT if: (i) the \t{M.READ} at \Cref{app:line:sg-write:parallel} returns in a single RTT and (ii) the \writeop{} takes the fast path. 
\begin{itemize}
  \item As discussed in \Cref{sec:wmr-complexity}, condition (i) happens if the write-back part of the \readop{} is not necessary, because a majority of the underlying max registers already contain the max value. This can occur if there is low contention, or if some majority of the underlying max registers is consistently more responsive than the rest of the max registers.
  \item Condition (ii) occurs if writer $p$ guesses a fresh timestamp, and no other process writes a higher timestamp before $p$ reads from \t{M}. Process $p$ guesses a fresh timestamp if the system has good clock synchrony.
\end{itemize}

A \readop{} returns in a single RTT if and only if (1) the \t{M.READ} on the first iteration of the while loop returns in 1 RTT, and (2) the value returned by this \readop{} is \t{VERIFIED}. Condition (1) occurs in the same scenarios as discussed above for (i). Condition (2) is true if (a) the previous \writeop{} performed the slow path \writeop{} at \Cref{app:line:sg-write:verified}, or (b) the previous \writeop{} performed the fast path and enough time has passed, such that its background \t{VERIFIED} \writeop{} has completed, or (c) the previous \writeop{} performed the fast path and, in the meantime, some \readop{} saw the \t{GUESSED} value and wrote it back as a \t{VERIFIED} value.

Overall, in the common case of low contention and good clock synchrony, every \readop{} and \writeop{} operation completes in a single RTT.

\end{document}